\pdfoutput=1
\RequirePackage[T1]{fontenc}
\documentclass[12pt]{amsart}

\usepackage[height=8.85in,width=6.45in]{geometry}

\usepackage[whole]{bxcjkjatype}
\usepackage[utf8]{inputenc}
\usepackage{amsmath}
\usepackage{amssymb}
\usepackage{mathtools}
\numberwithin{equation}{section}
\usepackage{slashed}
\usepackage{braket}
\usepackage[svgnames]{xcolor}
\usepackage[colorlinks,citecolor=DarkGreen,linkcolor=FireBrick,urlcolor=FireBrick,linktocpage,breaklinks=true]{hyperref}
\usepackage{cite}
\usepackage[pdftex]{graphicx}
\usepackage{tikz}
\usepackage{tikz-cd}
\usepackage{times}
\usepackage{courier}
\usepackage{bm}
\usepackage{subfig}
\usepackage{fnpct}
\usepackage{xcolor}
\usepackage{mdframed}

\usepackage{mathrsfs} 

\usepackage[all]{xy}
\usepackage{enumitem}

\renewenvironment{figure}[1][]{
  \begin{originalfigure}[#1]
    \begin{mdframed}[linecolor=black!0,backgroundcolor=black!1]
}{
    \end{mdframed}
  \end{originalfigure}
}

\theoremstyle{plain}
\newtheorem{defn}[equation]{Definition}
\newtheorem{thm}[equation]{Theorem}

\newtheorem{prop}[equation]{Proposition}
\newtheorem{fact}[equation]{Fact}
\def\PhysicsFact{Physics Assumption}
\newtheorem{physfact}[equation]{\PhysicsFact}
\newtheorem{fact?}[equation]{Fact?}
\newtheorem{cor}[equation]{Corollary}
\newtheorem{lem}[equation]{Lemma}

\newtheorem{conj}[equation]{Conjecture}

\theoremstyle{remark}
\newtheorem{rem}[equation]{Remark}
\newtheorem{ex}[equation]{Example}

\def\bC{\mathbb{C}}

\def\bR{\mathbb{R}}
\def\bZ{\mathbb{Z}}

\def\Hom{\mathrm{Hom}}
\def\tr{\mathop{\mathrm{tr}}}

\def\H{\mathrm{H}}
\def\K{\mathrm{K}}
\def\KO{\mathrm{KO}}

\def\MF{\mathrm{MF}}
\def\TMF{\mathrm{TMF}}
\def\Tmf{\mathrm{Tmf}}
\def\tmf{\mathrm{tmf}}

\def\pt{\mathrm{pt}}

\def\BO{B\mathrm{O}}
\def\Z#1{\mathbb{Z}/#1\mathbb{Z}}
\def\IZ{I\bZ}
\def\Q{\mathbb{Q}}

\def\ch{\mathrm{ch}}
\def\chd{\mathrm{chd}}
\def\Spin{\mathrm{Spin}}
\def\Stri{\mathrm{String}}
\def\SU{\mathrm{SU}}
\def\Nequals#1{$\mathcal{N}{=}#1$}
\def\cA{\mathcal{A}}
\def\pinm{\text{\upshape pin$-$}}
\def\Pinm{\mathrm{Pin}^-}
\def\spin{\text{\upshape spin}}
\def\stri{\text{\upshape string}}
\def\Wit{\mathop{\mathrm{Wit}}\nolimits}
\def\SQFT{\mathrm{SQFT}}
\def\GL{\mathrm{GL}}

\def\Anom{\text{\scshape Anom}}
\def\Het{\text{\scshape Het}}
\def\DefClass{\text{\scshape DefClass}}

\begin{document}

\title[TMF and heterotic global anomalies]{Topological modular forms and \\
the absence of all heterotic global anomalies}
\author{Yuji Tachikawa}
\author{Mayuko Yamashita}
\date{August, 2021. \emph{Revised}: October, 2022; January, 2023}
\address{Kavli Institute for the Physics and Mathematics of the Universe \textsc{(wpi)},
  The University of Tokyo,
  5-1-5 Kashiwanoha, Kashiwa, Chiba, 277-8583,
  Japan
}
\email{yuji.tachikawa@ipmu.jp}
\address{Research Institute for Mathematical Sciences, Kyoto University, 
Kita-shirakawa Oiwake-cho, Sakyo-ku, Kyoto, 606-8502, Japan
}
\email{mayuko@kurims.kyoto-u.ac.jp}
\thanks{\emph{Acknowledgments}: YT thanks helpful discussions with Arun Debray and Kantaro Ohmori during the early stage of this work,
and the authors thank Justin Kaidi, Kantaro Ohmori and Kazuya Yonekura for detailed comments on  the draft.
The authors also thank an anonymous referee for helpful inputs which led to important clarifications in the paper.
}

\begin{abstract}
We reformulate the question of the absence of global anomalies of heterotic string theory 
mathematically in terms of a certain natural transformation 
$\TMF^\bullet\to (I_\bZ\Omega^\text{string})^{\bullet-20}$,
from topological modular forms to the Anderson dual of string bordism groups,
using the Segal-Stolz-Teichner conjecture.
We will show that this natural transformation vanishes, implying that heterotic global anomalies are always absent.
The fact that $\TMF^{21}(\pt)=0$ plays an important role in the process.
Along the way, we also discuss how the twists of $\TMF$ can be described under the Segal-Stolz-Teichner conjecture,
by using the result of Freed and Hopkins concerning anomalies of quantum field theories.

The paper contains separate introductions for mathematicians and for string theorists,
in the hope of making the content more accessible to a larger audience.
The sections are also demarcated cleanly into mathematically rigorous parts and those which are not.
\end{abstract}
\maketitle

\tableofcontents

\section{Introduction}
\subsection{Introduction for mathematicians} 
We start this paper by an introduction for mathematicians. 
We provide a separate introduction for string theorists in Sec.~\ref{sec:introstrings},
but we urge string theorists to have a look at this introduction, at least to understand the notations.

\subsubsection{Motivations}
String theory is a theoretical framework where general relativity can be treated quantum mechanically.
It is not yet mathematically rigorously defined, 
but string theorists have already identified many places in its construction where it could become inconsistent,
due to the phenomena called anomalies.
When a non-vanishing anomaly is identified, string theory can be safely discarded.

Anomalies can be classified into the free part and the torsion part, usually called perturbative anomalies and  global anomalies in the physics literature.
The general theory of perturbative anomalies was developed in the 1980s,
and the absence of perturbative anomalies of string theory was firmly established around that time.
(On the physics side it was done in e.g.~\cite{AlvarezGaume:1983ig,Green:1984sg,Schellekens:1986xh,Lerche:1987qk}.
For mathematical reformulations, see e.g.~\cite{Liu:1994wh,Han:2012yt}.)
In contrast, global anomalies of string theory were studied only in a couple of references such as \cite{Witten:1985xe,Witten:1985mj,Witten:1985bt}.
One reason can be attributed to the absence of a general theory of global anomalies at that time.

Thankfully, the situation changed in the last decade; 
developments in the study of topological phases of matter in condensed matter physics 
eventually led us to a general mathematical framework to study global anomalies:
\begin{physfact}
\label{fact:anomaly}
The anomaly of a $d$-dimensional unitary quantum field theory defined on spacetimes with stable tangential structure $\mathcal{B}$ is
characterized by a $(d+1)$-dimensional invertible quantum field theory, which is in turn classified by \[
	(I_\bZ \Omega^\mathcal{B})^{d+2}(\pt).
\]
\end{physfact}
\noindent Here $\Omega^{\mathcal{B}}$ is the bordism homology with tangential structure  $\mathcal{B}$, represented by the Madsen-Tillmann spectrum $MT\mathcal{B}$,
and $I_\bZ\Omega^{\mathcal{B}}$ is its Anderson dual.%
\footnote{%
Here we make some historical remarks concerning Physics Assumption \ref{fact:anomaly}.
That the anomaly of a $d$-dimensional quantum field theory is specified by a $(d+1)$-dimensional invertible quantum field theory is long known, usually under the name \emph{anomaly inflow}.
This goes back implicitly to e.g.~\cite{Alvarez-Gaume:1984zst,Callan:1984sa,Faddeev:1984ung,Witten:1985xe} in the middle 1980s, 
but it is not clear who first formulated it clearly as a guiding principle to understand the anomaly in general.
We note that the concept of invertible quantum field theories was explicitly introduced in the literature in \cite{Freed:2004yc};
we also note that invertible quantum field theories are often simply called as invertible phases, which we follow.

The classification of such invertible phases was recognized as an important physics question in an influential paper by \cite{Chen:2011pg} where 
$\mathrm{H}^{d}(BG,U(1))$ was proposed to classify \emph{bosonic symmetry-protected topological phases},
which in the language of this paper are invertible phases defined on oriented manifolds.
That the classification was given by the Pontryagin dual of bordism groups was conjectured in \cite{Kapustin:2014dxa}.
It was proved then in \cite{FreedHopkins2021}
that the deformation classes of  topological invertible phases are classified by the Pontryagin dual of the torsion part of the bordism groups,
using the techniques of extended topological quantum field theories.
It was also conjectured  there that the classification of deformation classes of not-necessarily-topological invertible phases is given by the Anderson dual of the bordism groups.
That topological invertible quantum field theories are classified by the Pontryagin dual of the bordism groups was proved more elementarily in \cite{Yonekura:2018ufj} using non-extended topological quantum field theories in the sense of Atiyah,
and the conjecture of \cite{FreedHopkins2021} on the Anderson dual was put on a firmer footing in \cite{Yamashita:2021cao}
by constructing a model of the differential Anderson dual of bordism groups 
by directly formalizing what physicists use to describe invertible phases.
The anomalies are given by a deformation class of invertible quantum field theories,
and therefore are classified by the Anderson dual.
Lastly it should be mentioned that there is a nice set of lecture notes for mathematicians by Freed on these issues, see \cite{FreedLectures}. 
}
In this paper we will only deal with unitary quantum field theories,
and therefore the adjective `unitary' would often be dropped.

For the Anderson duals, we refer to \cite[Appendix B]{HopkinsSinger2005};
here we recall only the absolute basics.
We denote by $I\bZ$ the Anderson dual to the sphere spectrum. 
For a spectrum $E$, we denote by $I_\bZ E := F(E, I\bZ)$ its Anderson dual. 
The crucial property is that we have a natural exact sequence
\begin{align}\label{eq_exact_IE}
    0 \to \mathrm{Ext}(E_{d-1}(X), \bZ)  \to (I_\bZ E)^{d}(X) \to \Hom(E_{d}(X), \bZ) \to 0
\end{align}
for any spectrum $X$. 
With \PhysicsFact~\ref{fact:anomaly} in hand, we can now commence a systematic study of global anomalies of string theories.

In this paper, we concentrate on heterotic string theory, which is one of the variants of string theories. 
For the purpose of this paper, heterotic string theory can be considered as an elaborate machinery 
which produces a unitary $d$-dimensional quantum gravity theory from a given 
unitary two-dimensional \Nequals{(0,1)} superconformal field theory (SCFT) with certain suitable properties:
\begin{equation}
    \xymatrixcolsep{5pc}
        \xymatrix{
        \Het: \left\{\vcenter{\txt{\small 2d SCFTs with \\ \small $(c_L,c_R)=(26-d,\tfrac32(10-d))$}}\right\} 
        \ar[r]^-{\text{heterotic string}}_-{\text{construction}}  
        & \left\{
        \vcenter{\txt{\small $d$-dimensional quantum\\ 
        \small gravity theories}}
        \right\} },
\end{equation}
where $c_L$, $c_R$ specify the left- and right- central charges of the  input SCFT.

The resulting quantum gravity theory is known to be
defined on a manifold with string structure.\footnote{%
The string structure is associated to the fibration
\[
\xymatrix{B\mathrm{String}\ar[r] &  
B\Spin \ar[r]\ar[d]^{p_1/2}&  
B\mathrm{SO}\ar[r]\ar[d]^{w_2}& 
 B\mathrm{O}\ar[d]^{w_1}\\
&K({\mathbb Z},4) &
K({\mathbb Z}/2,2) &
K({\mathbb Z}/2,1) },
\]and is obtained by successively giving
an orientation to trivialize $w_1$,
a spin structure to trivialize $w_2$,
and a string structure to trivialize $p_1/2$. 
}
Therefore, according to \PhysicsFact~\ref{fact:anomaly},
its anomaly is characterized by an element in $I_\bZ\Omega^\text{string}(\pt)$.
Therefore there should be a map: \begin{equation}
    \xymatrixcolsep{9pc}
        \xymatrix{
        \Anom: \left\{
        \vcenter{\txt{\small $d$-dimensional quantum\\ 
        \small gravity theories}}\right\}
       \ar[r]^-{\txt{\tiny extract anomalies}}_-{\txt{\tiny using Physics Assumption~\ref{fact:anomaly}}}
        & 
          (I_\bZ\Omega^\text{string})^{d+2}(\pt).
        }
\end{equation}
What we would like to do is 
to study the composition $\Anom\circ\Het$ and show that it vanishes.

\subsubsection{Mathematical reformulation}

Let us try to reformulate our question mathematically.
For this purpose the conjecture of Segal, Stolz and Teichner \cite{oldSegal,Segal,StolzTeichner1,StolzTeichner2}
plays a crucial role:
\begin{conj}
\label{conj:SST}
Let $\SQFT_{-\nu}$ be the `space' of 2d unitary spin \Nequals{(0,1)} supersymmetric quantum field theories (SQFT) with anomaly $\nu\in (I_\bZ\Omega^\spin)^4(\pt) \simeq \bZ$. 
The sequence $\{\SQFT_\bullet\}$ forms an $\Omega$-spectrum, and agrees with $\TMF$,
the topological modular forms. 
\end{conj}
The deformation class of an SQFT, then, should define a class in $\TMF$, and therefore there is a map
\begin{equation}
    \xymatrixcolsep{6pc}
  \xymatrix{
        \DefClass: \left\{\vcenter{\txt{\small 2d  SQFT with \\ \small anomaly $\nu$ }}\right\} 
        \ar[r]^-{\txt{\tiny take deformation class}}_-{\txt{\tiny using Conjecture~\ref{conj:SST}}}
        &
        \TMF^{-\nu}(\pt)}.
\end{equation}
This will become useful as we will see shortly,
since the input for the heterotic string construction is 
an SCFT, which is a special version of an SQFT.

There are various physical and mathematical pieces of evidence behind the conjecture. 
One mathematical evidence is that there is a simpler version which can be mathematically formulated and proved rigorously:
\begin{thm}[=Theorem 1.1 of \cite{StolzTeichner1}\footnote{%
We note that in \cite{StolzTeichner1}, the spaces  $\mathrm{SQM}_\bullet$ were defined and shown to be homeomorphic to a particular model of the classifying spaces of $\KO_\bullet$, but that the $\Omega$-spectrum maps were not constructed or interpreted physically on the side of $\mathrm{SQM}_\bullet$.
For the interpretation of the $\Omega$-spectrum maps in $\TMF$, see e.g.~\cite[Sec.~2.5]{Gaiotto:2019gef}.
}]
\label{thm:SST}
Let $\mathrm{SQM}_{-\nu}$ be the `space' of 1d unitary  \pinm\  \Nequals1 supersymmetric quantum field theory (i.e.~ time-reversal invariant \Nequals1 supersymmetric quantum mechanics) 
with anomaly $\nu\in (I_\bZ\Omega^\pinm )^3(\pt) \simeq \bZ/8\bZ$.
The sequence $\{\mathrm{SQM}_\bullet\}$ forms an $\Omega$-spectrum, and agrees with $\KO$. 
\end{thm}
\noindent In this paper we assume the validity of Conjecture~\ref{conj:SST}.
For the details on $\TMF$ and its connective version $\tmf$, we refer the reader to the textbook \cite{TMFBook};
we also have a very brief summary in Appendix~\ref{sec:tmfdata}.

The importance of this conjecture for us stems from the following physics assumption,
which we will explain more fully below in this Introduction, and then in more detail in Sec.~\ref{subsec:reformulation}:
\begin{physfact}
\label{basic-assump}
The composition $\Anom\,\circ\,\Het$ factors through $\TMF^\bullet(\pt)$ via\begin{equation}
\alpha_\stri:\TMF^\bullet (\pt) \to (I_\bZ\Omega^\stri)^{\bullet-20}(\pt)
\end{equation}
so that there is the following commuting square:
\begin{equation}
	\label{bigsq}
    \vcenter{
    \xymatrixcolsep{8pc}
        \xymatrix{
        \left\{\vcenter{\txt{\small\upshape 2d SCFTs with \\ \small $(c_L,c_R)=(26-d,\tfrac32(10-d))$}}\right\} 
        \ar[r]^-{\txt{\tiny\Het}}  \ar[d]_-{\txt{\tiny\DefClass}}
        & \left\{
        \vcenter{\txt{\small \upshape $d$-dimensional quantum\\ 
        \small \upshape gravity theories}}
        \right\} \ar[d]^-{\txt{\tiny\Anom}} \\
        \TMF^{22+d}(\pt) \ar[r]_-{\alpha_\stri} &  (I_\bZ\Omega^\stri)^{d+2}(\pt),
        }
       }
\end{equation}
Here we used a physics fact that an SCFT with the central charge $(c_L,c_R)$ 
is an SQFT whose anomaly is $\nu=2(c_R-c_L)$,
to determine the degree of $\TMF$ in which $\text{\upshape DefClass}$ takes values in.
\end{physfact}
 
More generally, we need to consider situations where everything is parameterized over an auxiliary space $X$.
In this case the following corollary of Conjecture~\ref{conj:SST} plays an important role:
\begin{cor}
\label{cor}
The deformation classes of 2d unitary spin \Nequals{(0,1)} supersymmetric quantum field theories with anomaly $\nu\in (I_\bZ\Omega^\text{})^4(pt)\simeq \bZ$ parameterized by $X$ form the Abelian group \[
[X,\SQFT_{-\nu}] = \TMF^{-\nu}(X).
\]
\end{cor}

Then the anomalies of quantum gravity theories produced by heterotic string construction should be characterized by a natural transformation \begin{equation}
\alpha_\stri:\TMF^\bullet (X) \to (I_\bZ\Omega^\text{string})^{\bullet-20}(X).
\label{twisted}
\end{equation}
Here we note that, to be more precise, the degrees $\bullet$ need to incorporate the twists of $\TMF$ on $X$,
which we will come back to in Sec.~\ref{subsec:tmftwist}.
We also note that this natural transformation $\alpha_\stri$ should be regarded as a string-theoretic generalization of a natural transformation discussed in \cite[Sec.~9]{FreedHopkins2021}, \begin{equation}
\label{freefermion}
\alpha: \KO^{d-2}(X)  \to (I_\bZ \Omega^\spin)^{d+2}(X),
\end{equation}
which sends a $d$-dimensional massless fermion theory to the deformation class of its anomaly invertible field theory.\footnote{%
This transformation $\alpha$ \eqref{freefermion} also fits in a commuting square analogous to \eqref{bigsq}.
In this case, the upper left corner is the space of supersymmetric quantum mechanical systems on the worldlines, 
the upper right corner is the space of free fermionic theories of spin $\tfrac12$,
and the upper horizontal arrow is the process known as the second quantization,
which is one of the first things physicists learn in their textbooks on quantum field theories.
The vertical arrow on the left is to take the deformation class using Theorem~\ref{thm:SST},
and the vertical arrow on the right is to extract the anomalies.
We emphasize that the transformation $\alpha$ \eqref{freefermion} does \emph{not} vanish.
The vanishing of $\alpha_\stri$ \eqref{twisted} should be considered as one of many mysterious properties of string theory.
}

Let us now  describe $\alpha_\stri$ in \eqref{twisted} in more detail;\footnote{%
Physically, two sources of anomalies in the heterotic string constructions are known.
The first is the anomalies of massless fermions,
and the second is the anomalies carried by the $B$-field.
In principle there can be other sources of anomalies, carried by hitherto-unknown subtler degrees of freedom in the heterotic string constructions.
As string theory is not yet mathematically rigorous, we cannot be certain that such additional sources of anomalies do not exist. 
In this paper, we only consider the two known sources of anomalies;
we will see that their combination vanishes in a rather nontrivial manner.
Our results can then be considered as a strong indication that we already know all sources of anomalies.
}
this discussion will also serve as the explanation why physicists think that Physics Assumption~\ref{basic-assump} holds.
For this purpose, we need a few morphisms among spectra constructed in \cite{
AHR10}: \begin{equation}
    \vcenter{
        \xymatrix{
        MT\Stri \ar[r]^-{\Wit_\stri}\ar[d]^-{\iota} &
        \TMF\ar[d]^-{\sigma } \\
        MT\Spin  \ar[r]^-{\Wit_\spin}&
        \KO((q))
        }
       },
    \label{slwg}
\end{equation}
where  $MT\Stri$, $MT\Spin$ are the Madsen-Tillmann spectra representing $\Omega^\stri$ and $\Omega^\spin$ (in these cases they are weakly equivalent to the Thom spectra $M\Stri$ and $M\Spin$, respectively),
$\Wit_\stri$ is the string orientation of $\TMF$,
$\iota$ is the forgetful map,
and $\KO((q))$ is the Laurent series whose coefficients are $\KO$.
We note that the construction of $\sigma$ is more fully described in \cite[Appendix A]{HillLawson}.
We also note that in terms of Conjecture~\ref{conj:SST} and Theorem~\ref{thm:SST} of Segal-Stolz-Teichner,
the morphism $\sigma$ corresponds to putting a given 2d SQFT on $S^1$,
making it effectively an SQM system equivariant with respect to the rotation of $S^1$,
such that the power of $q$ specifies the weight under this action.

These are essential parts of the spectrum-level description of the Witten genus \cite{Witten:1986bf,Witten:1987cg}
in the following sense:
we have  a commutative diagram
    \begin{align}
    \vcenter{
        \xymatrix{
        \Omega^\stri_{\nu}(\pt) \ar[r]^-{\Wit_\stri}\ar[d]^-{\iota} &
        \TMF_{\nu}(\pt) \ar[r]^-{\phi}\ar[d]^-{\sigma } &  \MF[\Delta^{-1}]_{\nu/2} \ar[d]^-{\text{$q$-expansion}} \\
        \Omega^\spin_{\nu}(\pt) \ar[r]^-{\Wit_\spin}&
        \KO((q))_{\nu}(\pt) \ar[r]^-{\delta}& \bZ((q))
        }
       }.
        \label{phiW}
    \end{align}
where the left half was obtained by taking the homotopy groups of the commutative diagram \eqref{slwg},
$\MF[\Delta^{-1}]$ is the ring of weakly-holomorphic integral modular forms 
where the degree is given by the weight of the modular forms,
$\phi$ is the edge morphism of the descent spectral sequence which associates integral modular forms to $\tmf$
and weakly-holomorphic integral modular forms to $\TMF$,
and $\delta$ sends $\KO_{8k}(\pt)\simeq \bZ$ by the identity and $\KO_{8k+4}(\pt)\simeq \bZ$ (using the generator whose complexification is two times the square of the complex Bott element) by a multiplication by two.
The combined homomorphism from $\Omega^\stri_\nu(\pt)$ to $\bZ((q))$ is the original Witten genus \cite{Witten:1986bf,Witten:1987cg}.

Then, the  natural transformation $\alpha_\stri$ of our interest is the composition 
\begin{equation}
\alpha_\stri:\TMF^\bullet (X) \xrightarrow{\sigma}
\KO((q))^\bullet (X) \xrightarrow{\alpha_\spin}
 (I_\bZ\Omega^\spin)^{\bullet-20}(X) \xrightarrow{I_\bZ \iota}
 (I_\bZ\Omega^\text{string})^{\bullet-20}(X).
 \label{natural}
\end{equation}
of three natural transformations.
The first map $\sigma$ was introduced in \eqref{slwg}.
The second map $\alpha_\spin$, is a certain generalization of the homomorphism 
$\alpha:\KO^\bullet\to (I_\bZ\Omega^\spin)^{\bullet-4}$ we already saw in \eqref{freefermion}.
The third map $I_\bZ \iota$ is the Anderson dual to $\iota$.
Our question, the absence of global anomalies of quantum gravity theories obtained by heterotic string constructions, is whether this composition $\alpha_\stri=I_\bZ \iota\circ \alpha_\spin \circ \sigma$ vanishes.\footnote{%
More precisely, the anomalies of the fermions is the composition $\alpha_\spin\circ\sigma$,
and it is believed that the $B$-field coupling can be set up to cancel any anomalies in the kernel of $I_\bZ\iota$.
Therefore, the vanishing of $\alpha_\stri$ implies that the $B$-field anomalies and the fermion anomalies can be made to cancel.
} 

In a previous paper \cite{Tachikawa:2021mvw} by one of the authors (YT),
the vanishing of the particular case when $\bullet=24$ and $X=\pt$ was established by a direct computation,
using known properties of $\phi$.
The main result of this paper is that $\alpha_\stri$ always vanishes.
This is done by showing that this composition $\alpha_\stri$ is controlled by a single element in $I_\bZ \TMF^{-20}(\pt)$, 
which turns out to be zero. 
The property $\pi_{-21}(\TMF)=0$ is a crucial ingredient.

\subsubsection{Twists of $\TMF$ under the conjecture of Segal, Stolz, and Teichner}
\label{subsec:tmftwist}
To fully describe the physics question,
the degrees appearing in \eqref{twisted} and elsewhere 
are not simply valued in $\bZ$ but need to incorporate possible twists of $\TMF$ groups.
The twists of $\TMF$ were previously studied in \cite{ABG} from a homotopy-theoretic perspective,
where it was shown that $\TMF$ on $X$ can be twisted by $[X,\BO\langle0,\ldots,4\rangle]$,
where we denote the $n$-stage Postnikov truncation of
a path-connected space $Y$  by $Y\langle 0, \ldots, n\rangle$.

For our purposes, however, we have to analyze the twists of $\TMF$ in the context of the conjecture of Segal, Stolz and Teichner,
under which $\TMF^\bullet(X)$ is identified with the group of deformation classes of 2d \Nequals{(0,1)} SCFTs parameterized by $X$.
Under the conjecture, it is natural to identify the twists with the possible anomalies of 2d theories parameterized by $X$.\footnote{%
This point was already noted in \cite{Gukov:2018iiq,Johnson-Freyd:2020itv}, but
the comparison to \cite{ABG} was not made in these earlier references.
}
This is again the question answered by \PhysicsFact~\ref{fact:anomaly}, which says that 
the anomalies of such theories take values in 
\begin{equation}
(I_\bZ \Omega^\spin)^4(X).  
\end{equation}
If we fix a basepoint in $X$, we have $(I_\bZ \Omega^\spin)^4(X)=\bZ \oplus (\widetilde{I_\bZ \Omega^\spin})^4(X)$, 
where the summand $\bZ$ has already been identified with the degree of $\TMF$ in  Conjecture~\ref{conj:SST}.
The question is the interpretation of the second summand, the reduced part of the Anderson dual of the spin bordism group.

We will show in Appendix~\ref{sec:twists} that there is a natural isomorphism \begin{equation}
[X, \bZ\times \BO\langle0,\ldots,4\rangle] \simeq ({I_\bZ \Omega^\spin})^4(X)
\label{twisteq}
\end{equation}
for any CW-complex $X$, showing the consistency among three considerations, the description of twists in \cite{ABG},
the description of anomalies using \PhysicsFact~\ref{fact:anomaly},
and the Segal-Stolz-Teichner conjecture.
There is also a  version of the statement for $\KO(X)$, which will also be described and proved.

\subsubsection{Organization of the paper}
The main content of this paper is in Sec.~\ref{sec:reformulation} and Sec.~\ref{sec:proof}.
In Sec.~\ref{sec:reformulation}, we start from a physics description of the anomalies of heterotic string theories,
with the aim of specifying the natural transformation \eqref{natural} in a mathematically well-defined language.
A mathematically-oriented reader will be able to understand more and more as s/he reads this section.
We then proceed in Sec.~\ref{sec:proof} to show that this natural transformation vanishes.

We have four appendices: in Appendix~\ref{sec:bordismdata}, we present the spin and string bordism groups up to $d=16$,
and in Appendix~\ref{sec:tmfdata}, we provide the abelian groups $\pi_{\bullet}(\TMF)$,
which is 576-periodic.
Both are known facts in the literature, but we provide them for the convenience for the readers.
Appendix~\ref{sec:twists} contains the proof of the statement \eqref{twisteq}
which is needed in Sec.~\ref{sec:reformulation} to translate the physics question into mathematics.
Finally in Appendix~\ref{sec:last}, we determine the homomorphism 
$(I_\bZ\Omega^{\mathcal{B}})^d(\pt)\to 
(I_\bZ\Omega^{\mathcal{B}'})^{d'}(\pt)$ associated to transformations of tangential structures in a few cases,
which will be needed in other parts of the paper.

\subsection{Introduction for string theorists} 
\label{sec:introstrings}
\subsubsection{The aim}
Let us now present an introduction for string theorists.
It is well-known that the 10d heterotic string theory is free of perturbative anomalies
thanks to the Green-Schwarz cancellation \cite{Green:1984sg}.
To ensure that heterotic string theory is fully consistent,
this needs to be generalized in two directions.
One is to study compactifications, and another is to consider global anomalies.

Let us start with the first one. 
As long as we compactify on a smooth manifold, the absence of anomalies in 10d guarantees that the anomalies are still absent in lower dimensions.
In heterotic string theory, however, you can use arbitrary 2d \Nequals{(0,1)} superconformal field theories (SCFTs) of central charge $(c_L,c_R)=(26-d,\tfrac32(10-d))$ to describe compactifications down to $d$ dimensions,
and these SCFTs might not come from the quantization of strings moving in a smooth geometry.
In \cite{Schellekens:1986xh,Lerche:1987qk,Lerche:1988np}, 
perturbative anomalies of such general compactifications were analyzed,
and were shown to vanish always via the Green-Schwarz mechanism.

As for the second question, 
it was shown in \cite{Witten:1985bt} that the 10d $E_8\times E_8$ heterotic string has no global anomalies.\footnote{%
The anomaly cancellation of 10d $\mathfrak{so}(32)$ string theory was analyzed in \cite{Freed:2000ta} from the Type I perspective.
}
It then follows that the global anomalies do cancel in the compactifications of $E_8\times E_8$ strings on smooth geometries.
But again, we need to analyze the global anomalies of heterotic compactifications using internal SCFTs which do not necessarily correspond to smooth internal manifolds.

In a previous paper \cite{Tachikawa:2021mvw} by one of the authors (YT),
a particular case of the $\bZ_{24}$ global anomaly of the $B$-field gauge transformation of heterotic compactifications down to two dimensions was analyzed.
The objective of this paper is to establish that global anomalies of all types (gravitational, gauge, or mixed) are absent in heterotic compactifications to arbitrary dimensions.

We pause here to mention that we only consider the anomalies of fermions and of the $B$-field,
and that we do not entertain the possibility of the anomalies carried by other unknown subtler degrees of freedom, for example an almost decoupled topological field theory 
produced by the heterotic string construction.
We will see that the anomalies of fermions and of the $B$-field cancel in a rather nontrivial manner, only after using rather sophisticated techniques from algebraic topology.
We take this as a strong indication that there is indeed no sources of anomalies in addition to the two known ones.

\subsubsection{The strategy}
The daunting task of analyzing arbitrary heterotic compactifications is made tractable by the following two observations:

\medskip

\noindent1.2.2.1. \emph{The use of $\TMF$.}
The anomaly is a discrete quantity, and therefore is independent of continuous deformations.
According to the conjecture of Stolz, Segal and Teichner  \cite{Segal,StolzTeichner1,StolzTeichner2}, 
the equivalence classes  under continuous deformations of 2d \Nequals{(0,1)} supersymmetric quantum field theories
having the same anomaly as $\nu$ chiral multiplets
form an Abelian group denoted by $\TMF^{-\nu}(\pt)$, known as the topological modular forms.\footnote{%
The same group  can be written in multiple ways: $\TMF^{-\nu}(\pt)=\TMF_\nu(\pt)=\pi_\nu(\TMF)=\pi_{0}(\TMF_{-\nu})$.
Note that this only holds when $X=\pt$.
}
More generally, the deformation classes of such theories parameterized by $X$ form an Abelian group $\TMF^{-\nu}(X)$.
These are presented as Conjecture~\ref{conj:SST} and Corollary~\ref{cor} in the introduction for mathematicians.
There are various physical and mathematical pieces of evidence behind this conjecture;
we discuss some of them slightly later in Sec.~\ref{sec:tmf-for-physicists}, 
so as not to disrupt the flow of the discussions here.

Therefore, by assuming the validity of this conjecture, we can employ various properties mathematicians uncovered for $\TMF$.
In particular, an SCFT $T$ with $(c_L,c_R)=(26-d,\frac32(10-d))$ determines a class $[T]\in \TMF^{22+d}(\pt)$.
This trick was already used in \cite{Tachikawa:2021mvw}.

An important subtlety here is that the continuous deformations determining a $\TMF$ class 
is \emph{not} required to preserve that the theory is conformal;
any deformation which preserves supersymmetry is allowed.
In contrast, the heterotic string construction itself requires an SCFT as an input, not a general SQFT.
However, the fermion anomaly is determined by the spectrum of massless spacetime fermions, which can be read off from the right-moving R-sector ground states of the input 2d SCFT.
There is no problem in formally applying the same algorithm to the right-moving R-sector vacuum of any 2d SQFT
to obtain the spectrum of massless spacetime fermions,
from which the anomaly can also be computed.
As the anomaly is a discrete quantity,
it cannot change under continuous deformations.
Therefore, we can use a $\TMF$ class to compute the anomaly of the spacetime fermions.

\medskip

\noindent1.2.2.2 \emph{The reduction to the single case $d=-1$.}
There is in fact no need to consider different choices of the spacetime dimension $d$ and the gauge group $G$ separately.
To motivate this, let us recall that the global anomaly of a $d$-dimensional theory is captured by the invertible anomaly theory $\cA$ in $(d+1)$ dimensions.
Note in particular that the global anomaly in the traditional sense, associated to a global gauge transformation and/or a diffeomorphism which we collectively denote by $\phi$ on $M_d$, 
is given by the value the invertible phase $\cA$ assigns to the mapping torus $N_{d+1}$,
which is obtained by taking $[0,1] \times M_d$ and gluing the two ends by the transformation $\phi$.

This means that, to compute the anomalous phase for a heterotic compactification to $d$ dimensions,
we can equivalently make a further compactification on $M_d$ first to obtain a zero-dimensional compactification,
and compute its anomalous phase under $\phi$,
where we can regard that $\phi$ is an element in some symmetry group $G$.
In short, we can reduce the case to $d=0$, but with an arbitrary symmetry group $G$.

There are two obstacles in putting this observation in practice: \begin{enumerate}
	\item To perform a further heterotic compactification on $M_d$,
we need to arrange various fields on it to solve the  equations of motion of the heterotic string theory.
But there is no guarantee that the configuration $M_d$ detecting the anomaly can be arranged to allow such  a solution.
	\item The modern understanding of anomalies requires us to consider $N_{d+1}$ not necessarily of the form of a mapping torus.
This formally requires us to consider a `compactification to $-1$ dimensions' and evaluate its $0$-dimensional anomaly theory on a point.
\end{enumerate}
What saves us from these problems is that our analysis does not directly use the worldsheet SCFTs
but only the associated $\TMF$ classes.
These mathematical objects are less rigid and more flexible than SCFTs, 
and we can make perfect sense of compactification on $M_d$ or $N_{d+1}$.
An added bonus is that, once we compactify using $N_{d+1}$, we do not have to consider any symmetry group, since the spacetime is now a single point.
This allows us to reduce our entire question to the case $d=-1$ and without any symmetry $G$.
For this we only need to study $\TMF^{22+d}(\pt)=\TMF^{21}(\pt)$, which is known to be trivial.
This implies that the global anomaly is always absent.\footnote{%
In an email to one of the authors (YT) while this draft was being prepared, 
Edward Witten independently noted that the vanishing of the global anomaly of heterotic strings follows if $\TMF^{21}(\pt)$ vanishes.
}\footnote{%
The authors also would like to warn the reader that in the following sections what is outlined in this introduction for physicists is not going to be directly implemented rigorously in mathematics.
Rather, what we do is to formulate the anomaly question in the language of algebraic topology using morphisms between spectra, 
with which one can still reduce the question to the vanishing of $\TMF^{21}(\pt)$.
To actually implement what is indicated in this introduction for physicists, we would need to deal with a differential version of $\TMF$.
}

\subsubsection{Some features of $\TMF$ for physicists}
\label{sec:tmf-for-physicists}
Before proceeding, let us discuss some features of $\TMF$ 
which hopefully would help string theorists to understand its relationship to quantum field theories.
We note that a nice review  for string theorists was already given in \cite{Gukov:2018iiq},
and  some physics checks of this conjecture were given in \cite{Gaiotto:2018ypj,Gaiotto:2019asa,Gaiotto:2019gef,Johnson-Freyd:2020itv};
the readers are also recommended to consult these references.

First, let us consider 2d \Nequals{(0,1)} sigma model whose target space is an $n$-dimensional manifold $M_n$.
The right-moving fermions in general have sigma model anomalies \cite{Moore:1984dc,Moore:1984ws,Manohar:1984zj}, 
which can be trivialized if and only if a $B$-field on $M_n$ can be specified so that its field strength $H$ satisfies ``$dH=\tr R^2$\,'' \cite{Witten:1985mj}.
The resulting theory has a gravitational anomaly which is $n$ times that of a chiral multiplet.

According to Conjecture~\ref{conj:SST}, the deformation classes of such theories define an element in $\TMF_n(\pt)$.
The manifolds equipped with $B$-fields satisfying ``$dH=\tr R^2$\,'' are called string manifolds by mathematicians,
and it was found in \cite{Gaiotto:2019asa} that,
when two string manifolds are bordant,
the resulting sigma models can be connected by going up and down along the renormalization group flow.

This means that there should be a commuting square
\begin{equation}
\label{sigma}
    \vcenter{
    \xymatrixcolsep{8pc}
        \xymatrix{
        \left\{\vcenter{\txt{\small $n$-dimensional manifold $M_n$ \\ \small with $dH=\tr R^2$ }}\right\} 
        \ar[r]^-{\Sigma}  \ar[d]_-{\text{take bordism class}}
        & \left\{
        \vcenter{\txt{\small 2d \Nequals{(0,1)} QFT with\\ 
        \small $n$ units of gravitational anomaly}}
        \right\} \ar[d]^-{\text{take deformation class}} \\
       \Omega^\text{string}_n(\pt)  \ar[r]_-{\Wit_\stri} &  \TMF_n(\pt)
        }
       }.
\end{equation}
Here, $\Sigma$ on the upper horizontal arrow 
is the operation creating a sigma model on the manifold by the path integral,
which is yet to be rigorously defined,
and $\Wit_\stri$ on the lower horizontal arrow
was mathematically constructed in \cite{AHS1,AHS2,AHR10}
and is sometimes known as the string orientation.

Second, given a 2d \Nequals{(0,1)} quantum field theory $T$, we can consider its elliptic genus $Z_\text{ell}(T;q)$
in physicists' sense, which is the Witten index in the right-moving R-sector.
This is almost modular invariant but not quite, due to the gravitational anomaly. 
This can be cancelled by multiplying by $\eta(q)^n$.
We can then define the Witten genus of the theory $T$ by $
\eta(q)^n Z_\text{ell}(T;q)
$,
which is a weakly-holomorphic modular form of weight $n/2$;
here the adjective weakly-holomorphic means that one allows poles at $q=0$.
The Witten genus is independent of continuous deformations, and it descends to a map defined on $\TMF$,
resulting in another commuting square
\begin{equation}
\label{phi}
    \vcenter{
    \xymatrixcolsep{8pc}
        \xymatrix{
        \left\{
        \vcenter{\txt{\small 2d \Nequals{(0,1)} QFT with\\ 
        \small $n$ units of gravitational anomaly}}
        \right\} \ar[d]^-{\text{take deformation class}} 
        \ar[r]^-{\eta(q)^n Z_\text{ell}(-;q) }  & \MF[\Delta^{-1}]_{n/2}\ar@{=}[d]  \\
         \TMF_n(\pt) \ar[r]^-{\phi} & \MF[\Delta^{-1}]_{n/2}
        }
       }.
\end{equation}
Here, $\MF$ is the ring of modular forms with integer $q$-expansion coefficients,
and the notation $\MF[\Delta^{-1}]$ means that we allow inverting the modular discriminant $\Delta$ which has a first-order zero at $q=0$, resulting in the ring of weakly-integral modular forms.
Again, the upper horizontal arrow is not rigorously defined,
but the lower horizontal map $\phi$ is well-defined and has been completely determined by mathematicians \cite{Hopkins2002}.

We can now combine the two commuting squares \eqref{sigma} and \eqref{phi}.
Then, $\phi\circ \Wit_\stri ([M_n,B])$ computed on the lower horizontal arrows 
should equal $\eta(q)^n$ times the elliptic genus of the sigma model whose target space is $M_n$ with the specified $B$-field, 
which physicists know how to compute \cite{Witten:1986bf}.
What mathematicians constructed reproduces this physics expectation.

We also note that the map $\Wit_\stri$ which appeared in \eqref{sigma}, 
which should be physically understood as the quantization map, can be generalized further.
For this purpose one considers a fibration $F\to E\xrightarrow{p} B$ so that  it is  equipped with a fiber-wise $B$-field
solving ``$dH=\tr R^2$\,'', i.e.~a string orientation.
Let us now consider a family of 2d \Nequals{(0,1)} theories with $m$ units of gravitational anomalies parameterized by the total space $E$,
which specifies a class in $\TMF^{-m}(E)$.
In this setup, the morphism of ring spectra $\Wit_\stri \colon MT\mathrm{String} \to \TMF$ (the string orientation of $\TMF$) gives us the pushforward map \begin{equation}
p_! : \TMF^{-m}(E) \to \TMF^{-m-\dim F}(B).
\end{equation}
Mathematically, this generalizes the integral of cohomology classes along the fiber, reducing the degree of the class by the dimension of the fiber.
Physically, this operation performs a quantization along the fiber $F$,
so that we have a family of 2d \Nequals{(0,1)} theories parameterized by the base $B$,
now with $m+\dim F$ units of gravitational anomaly.
The original version \eqref{sigma} is obtained by taking the fibration $M\to M\to \pt$ and by considering the pushforward $p_! : \TMF^0(M) \to \TMF^{-n}(\pt)=\TMF_n(\pt)$.
Then $\Wit_\stri([M])=p_!(1)$, where $1\in \TMF^0(M)$ represents the trivial constant family.

\subsubsection{Organization of the paper}
The rest of the paper is organized as follows.
In Sec.~\ref{sec:reformulation}, we start from a physics description of the anomalies of heterotic string theories,
and translate it into a mathematical setup where the full machinery of algebraic topology can be effectively employed.
A physics-oriented reader will feel more and more alien as s/he reads this section.
The objective of Sec.~\ref{sec:proof} is then to show rigorously the vanishing of anomalies
in that formulation, which will turn out to be relatively straightforward after all the preparations done in Sec.~\ref{sec:reformulation}.

We also have four appendices: in Appendix~\ref{sec:bordismdata}, we present the spin and string bordism groups up to $d=16$,
and in Appendix~\ref{sec:tmfdata}, we provide the abelian groups $\TMF_\bullet(\pt)$,
which is 576-periodic.
Both are known facts in the literature, but we provide them for the convenience for the readers.
Appendix~\ref{sec:twists} contains a mathematical discussion of how to reconcile 
the understanding of the crucial relation ``$dH=\tr R^2-\tr F^2$\,''
from three points of view, namely from algebraic topology, from worldsheet, and from spacetime.
In Appendix~\ref{sec:last}, we discuss a few examples of how the anomalies are translated when we change the symmetry structure; 
the results are needed in other parts of the paper.

The rest of the paper utilizes various notions from algebraic topology.
String theorists who would like to learn them would find ample explanations 
e.g.~in the textbooks \cite{Stong1968,Rudyak}, in the review \cite{BCguide},  or the lecture note \cite{FreedLectures}.

\section{Description of  anomalies}
\label{sec:reformulation}

In this section we start from the physics description of anomalies of heterotic compactifications and translate it to a certain natural transformation from $\TMF$ to $I_\bZ\Omega^\text{string}$.
What we do here is to explain why physicists think the main Physics Assumption~\ref{basic-assump} is reasonable,
by explaining its rationale and decomposing it into a number of more basic Physics Assumptions,
which we summarize at the end of this section in Sec.~\ref{subsubsec_physfacts} as Physics Assumptions~\ref{physicsfact}, \ref{fact_rationalization}, \ref{fact_natural}.

\goodbreak

\subsection{Physics setup}
\subsubsection{Input}
We consider heterotic compactifications down to $d$ spacetime dimensions.
The internal degrees of freedom are described by a 2d \Nequals{(0,1)} SCFT $T$ whose central charge is given by $(c_L,c_R)=(26-d,\tfrac32(10-d))$.
When the theory $T$ has the symmetry $G$, the spacetime theory in dimension $d$ has $G$ as a gauge symmetry.
When the theory $T$ has the space $X$ of exactly marginal couplings, the space $X$ appears as the target space of the massless scalar fields of the spacetime theory.

Let us first consider the case when $G$ is a simply-connected simple compact Lie group,
and denote the level of the $G$ current algebra of the theory $T$ by $k\in \bZ$.
The generalization to arbitrary $G$ will be performed while we formulate the question more mathematically.

In principle, the theory $T$ can have an anomaly depending on the space $X$ of exactly marginal couplings as discussed in \cite{Tachikawa:2017aux,Cordova:2019jnf,Cordova:2019uob},
which would also affect our discussions in the following.
The generalization we perform later from simply-connected simple Lie groups to general groups in fact takes care of the effect of having the space $X$ of exactly marginal couplings,
so we will neglect the effect of $X$ for the moment.

\subsubsection{Spacetime}
The $d$-dimensional spacetime $M$ is equipped with an orientation, a metric, a spin structure, 
a $G$-bundle with connection,
and then a $B$-field, whose gauge-invariant field strength $H$ satisfies
the relation roughly of the form ``
$dH=\tr R^2-k\, \tr F^2$
''
where $R$ is the spacetime curvature and $F$ is the curvature of the $G$ gauge field.
Mathematically, this set of fields determine a string structure on $M$
twisted by the pullback of the class $k\in \bZ=\H^4(BG,\bZ)$ via the classifying map $M\to BG$.

\subsubsection{Massless fermions}
We are primarily interested in the massless fermion fields in the spacetime theory.
They arise from the worldsheet theory $T$ by putting the right-movers to the R-sector vacuum.
Then, the lowest modes of the left-movers (with $L_0=0$) give rise to massless gravitinos and dilatinos,
and the first excited states of the left-movers (with $L_0=1$) give rise to other massless spin-$\tfrac12$ fermions.
In even dimensions, the spacetime chirality is correlated with the right-moving fermion number $(-1)^{F_L}$ via the GSO projection.

\subsubsection{Anomalies}
We  then need to compute their anomalies.
The anomalies, both perturbative and global, of spin-$\tfrac12$ fermions and gravitinos 
were worked out in \cite{AlvarezGaume:1983ig,Witten:1985xe,Witten:1985mj},
so we can simply quote them and sum over them.
On general spin manifolds, the resulting total fermion anomaly does not usually vanish.
What we would like to ask is whether it vanishes when restricted to the manifolds satisfying the relation `` $dH=\tr R^2-k\tr F^2$ ''.
When it does not, this means that there is a residual anomaly even after the Green-Schwarz anomaly cancellation mechanism is used.
The perturbative case was already settled in the 1980s by \cite{Schellekens:1986xh,Lerche:1987qk,Lerche:1988np}.
Our main question concerns the global anomalies.

\subsection{Mathematical reformulation}
\label{subsec:reformulation}
Let us reformulate the physics description in a more mathematically palatable language. 

\subsubsection{$\TMF$ class as the input}

Massless spacetime fermions of a heterotic string constructed from an SCFT $T$
come from the R-sector 
vacuum of the right-movers;
we will neglect the effect of $G$ for a while to simplify the presentation.
The right-moving R-sector vacuum is a particular subspace $\mathcal{H}_T$ of the Hilbert space of the theory $T$ on the circle;
as such $\mathcal{H}_T$ has an action of $S^1$.
We use powers of $q$ to grade the $S^1$ action,
and physicists' convention is to use $q^{L_0-c_L/24}$ to  grade the pieces,
where $L_0$ is the 0th Virasoro generator and $c_L$ is the left-moving Virasoro central charge.
Each graded piece, i.e.~each eigenspace of $L_0$, can in general jump under continuous deformations,
but it determines a class in $\KO^{-\nu}(\pt)$,
where $\nu=2(c_R-c_L)$.
This is because each piece is a supersymmetric quantum mechanical system  with gravitational anomaly labeled by $\nu$, which according to Theorem~\ref{thm:SST} defines a class in $\KO^{-\nu}(\pt)$;
the time reversal operator is given by the CPT operator of the original 2d theory.

Combining the graded pieces, we see that $\mathcal{H}_T$ 
defines a class $[\mathcal{H}_T]\in \KO^{-\nu}((q^{1/24}))(\pt)$
of the form  \begin{equation}
[\mathcal{H}_T]=q^{-c_L/24}(V + W q+\cdots) \in \KO^{22+d}((q^{1/24}))(\pt).
\label{AB}
\end{equation}
Here, $V$ gives dilatinos and gravitinos,
while $W$ gives other spin-$\tfrac12$ fermions,
and the fermion anomalies of a heterotic string construction
can be determined from $[\mathcal{H}_T]$ without the detailed knowledge of $T$.

Luckily for us, the map $T\mapsto [\mathcal{H}_T]$ factors through $\TMF$, since we have 
\begin{equation}
\sigma([T])= \eta(q)^\nu [\mathcal{H}_T],
\label{3.4}
\end{equation} where  $[T]\in \TMF^{-\nu}(\pt)$ is the TMF class of $T$, 
$\eta(q)$ is the Dedekind eta,
and  we already encountered $\sigma: \TMF^{-\nu}\to \KO^{-\nu}((q))$ in \eqref{phiW}.\footnote{%
The extra factor of $\eta(q)^\nu$ corresponds to adding $\nu$ left-moving fermions to cancel the worldsheet anomaly.}
We can then determine the fermion anomalies from $[T]\in \TMF^{-\nu}(\pt)$ 
without the detailed knowledge of $T$.

This allows us to take $\TMF$ classes $[T]$ as an input,
rather than SCFTs $T$,
for the purpose of our analysis,
giving us considerable flexibility:
the heterotic string constructions require an SCFT with a fixed $(c_L,c_R)$ as the input.
But according to Conjecture \ref{conj:SST} of Segal-Stolz-Teichner, 
the equivalence classes which defines a TMF class is much more relaxed.
The allowed deformations do not even have to preserve conformality,
as long as \Nequals{(0,1)} supersymmetry is preserved,
so that we can isolate the supersymmetric R-sector vacuum of the right-movers.
The gravitational anomaly $\nu=2(c_R-c_L)$ of $T$, which is known to be an integer, is preserved under such deformations.

When we further impose that the theory $T$ has a symmetry $G$, 
its 't Hooft anomaly $k$ is also preserved under continuous deformations.
According to \PhysicsFact~\ref{fact:anomaly}, the anomaly of a 2d spin theory $T$ with symmetry $G$ is classified by 
\begin{equation}
(I_\bZ\Omega^\spin)^4(BG)
= \bZ \oplus (\widetilde{I_\bZ\Omega^\spin})^4(BG),
\end{equation}
and we regard $\nu\oplus k$ to be the element in this group.

Let us discuss $\nu$ and $k$ in turn. 
As for $\nu\in (I_\bZ\Omega^\spin)^4(\pt)=\bZ$, 
it is given by $\nu=2(c_L-c_R)=-22-d$.
As for $k$,
when $G$ is a simple simply-connected compact Lie group, 
we have\footnote{%
This follows easily from the Atiyah-Hirzebruch spectral sequence using $\H^d(BG,\bZ)=\bZ,0,0,0,\bZ$ for $d=0,1,2,3,4$.
It also follows from our Proposition~\ref{prop_twist}.
} $(\widetilde{I_\bZ\Omega^\spin})^4(BG)=\H^4(BG,\bZ)=\bZ$,
and $k$ can be further identified with an integer.
We can  then formulate the equivariant version of the Segal-Stolz-Teichner conjecture:
\begin{conj}
The equivariant $\TMF$ group $\TMF^{-(\nu\oplus k)}_{G}(\pt)$ is the deformation class of 2d \Nequals{(0,1)} supersymmetric theories 
with the anomaly $\nu\oplus k\in (I_\bZ\Omega^\spin)^4(BG)$.
\end{conj}
\noindent In our case,  the class $[T]\in \TMF^{-(\nu\oplus k)}_{G}(\pt)$ should allow us to determine the fermion anomaly of the spacetime theory,
via 
$\sigma:\TMF^{-(\nu\oplus k)}_{G}(\pt)\to \KO^{-(\nu\oplus\tilde k)}_{G}((q))(\pt)$,
where $\tilde k$ is the twist of $\KO$ naturally induced by the twist $k$ of $\TMF$.

More generally, we expect the following conjecture to hold: 
\begin{conj}
The twisted $\TMF$ group $\TMF^{-(\nu\oplus k)}(X)$ is the deformation class of 2d \Nequals{(0,1)} supersymmetric theories parameterized by $X$
with the anomaly $\nu\oplus k\in (I_\bZ\Omega^\spin)^4(X)$. 
\end{conj}
\begin{rem}
For the two conjectures above to make sense,
we need to show that elements of $(I_\bZ\Omega^\spin)^4(X)$ can be 
used as degrees and twists of $\TMF^\bullet(X)$,
as we already pointed out in Sec.~\ref{subsec:tmftwist}.
We will establish this fact as Proposition~\ref{prop_twist} later in Appendix~\ref{sec:twists}.
\end{rem}

\subsubsection{Spacetime tangential structure}

Let us describe the structure we need on the $d$-dimensional spacetime manifold $M_d$,
when we use $[T]\in \TMF^{22+d\oplus -k}_{G}(\pt)$ as the input.
We first send $[T]$ to the Borel equivariant version via $c:\TMF^\bullet_G(\pt)\to \TMF^\bullet(BG)$.
Now, $M_d$ is equipped with a $G$-bundle with connection, and in particular with a map $f:M_d\to BG$.
We can then pull back the class $c[T]$ to consider $f^*c[T]\in \TMF^{22+d\oplus-f^*(k)}(M_d)$.

The heterotic string theory construction demands that 
the motion of strings on $M_d$ should be consistently quantizable,
and the family of 2d \Nequals{(0,1)} quantum field theories specified by $f^*c[T]$ over $M_d$ 
should give an element of $\TMF^{22}(\pt)$ after the motion along $M_d$ is also quantized. 
Mathematically, as we discussed at the end of Sec.~\ref{sec:tmf-for-physicists},
this operation should be given by the pushforward map in $\TMF$, \begin{equation}
p_! : \TMF^{22+d\oplus-f^*(k)}(M_d) \to \TMF^{22}(\pt),
\end{equation}
associated to the projection $p: M_d\to \pt$.
This means that  $M_d$ needs to be equipped with a twisted string structure
realizing this shift in the degrees.
This in turn means that $M_d$ should be equipped with a string structure twisted by $-f^*(k)$,
which is naturally realized if 
$M_d$ together with the $G$-bundle defines a class in $[M_d,f]\in \Omega^\text{string}_{d\oplus -k}(BG)$.\footnote{%
The description of the spacetime structure for heterotic string theory in terms of twisted differential string structure 
was discovered and has been extensively developed by H. Sati, U. Schreiber and their collaborators. 
See e.g.~\cite{Sati:2009ic}.
}

Let us check  what this condition means when $G$ is a simply-connected simple compact Lie group.
We first note that elements 
$k\in [X,\BO\langle0,\ldots,4\rangle]$ 
can be used to provide twists of $\Omega^{\stri}_\bullet(X)$,
as we recall in Appendix~\ref{sec:twists}.
We also refer the reader to the same Appendix for the definition of the Postnikov truncation $Y\langle0,\ldots,n\rangle$.
More explicitly, the string structure twisted by $k$ on a manifold $M$ equipped with a map $f:M\to X$
is specified as follows:
we consider  the classifying map $g: M\to \BO$ of the tangent bundle,
and consider its projection $g\in [X,\BO\langle0,\ldots,4\rangle]$, which we denoted by the same symbol.
We now take the product $g \cdot  (k\circ f) \in [M,\BO\langle0,\ldots,4\rangle]$,
and we specify the trivialization of this product.
This is the twisted string structure.

Now let  $G$ is a simply-connected simple compact Lie group
and consider an element  $k\in [BG,\BO\langle0,\ldots,4\rangle]$.
The pull-back of $w_1,w_2 \in \H^{1,2}(\BO\langle0,\ldots,4\rangle,\bZ/2)$ is trivial,
and therefore  $k$ can be lifted uniquely to an element $k\in [BG,B\Spin\langle0,\ldots,4\rangle]$.
We now consider the generator $\lambda$ of $\H^4(B\Spin\langle0,\ldots,4\rangle,\bZ)\simeq \bZ$,
and take its pullback $\kappa=k^*(\lambda) \in \H^4(BG,\bZ)\simeq \bZ$.
The mapping from $k$ to $\kappa$ gives the isomorphism $[BG,\BO\langle0,\ldots,4\rangle]\simeq \H^4(BG,\bZ)$.

The string structure on $M$ with a map $f:M\to BG$ twisted by $-k$ is then the trivialization of $g \cdot f^*(\kappa)$.
This means that $M$ is a spin manifold, and that we have specified a classifying map $g: M\to B\Spin$,
and that $g^*(\lambda) - f^*(k^*(\lambda)) \in \H^4(M,\bZ)$ is trivialized.
This last condition is indeed what we have in string theory,
``$dH=\tr R^2 -k \tr F^2$\,'',
which is the famous Green-Schwarz condition.

\subsubsection{Anomalies of gravitinos and spin-$1/2$ fermions}
As described above, the spacetime theory is defined on $d$-dimensional spacetimes $M$ equipped with a twisted string structure.
Then \PhysicsFact~\ref{fact:anomaly} says that the anomalies to be computed takes values in $(I_\bZ\Omega^\text{string})^{d+2\oplus -k}(BG)$.
Therefore, we need to produce a homomorphism \begin{equation}
\alpha_\stri \colon \TMF^{22+d \oplus -k}_{G}(\pt) \to (I_\bZ\Omega^\text{string})^{d+2\oplus -k}(BG).
\end{equation}
To study the massless fermions, we only need the twisted spin structure,
and therefore we need to study a homomorphism \begin{equation}
\alpha_\spin:\KO^{22+d\oplus -\tilde k}_{G}((q))(\pt) {\to} (I_\bZ\Omega^\spin)^{d+2\oplus -\tilde k}(BG)
\end{equation} describing the anomalies of various massless fermions.\footnote{%
The authors appreciate Kantaro Ohmori for helpful discussions on the formulation of $\alpha_\spin$.
}
For a while, we neglect the effect of the twists, as it can be straightforwardly incorporated afterwards.

When $T$ is an SCFT of $(c_L,c_R)=(26-d,\tfrac32(10-d))$ and therefore $\nu=2(c_R-c_L)=-22-d$, the element \begin{equation}
\sigma([T]) = \eta(q)^{-22-d} q^{-(26-d)/24}(V+Wq +\cdots)\in \KO^{-\nu}((q))(\pt)
\label{special}
\end{equation}
has poles of order at most 2, by combining \eqref{AB}, \eqref{3.4} and $q^{\nu/24}q^{L_0-c_L/24}=q^{L_0-2}$.
2. Equivalently, we have \begin{equation}
\sigma([T]) \in q^{-2}\KO^{-\nu}[[q]](\pt).
\end{equation}

For our purposes, $\alpha_\spin$ needs to have the property that it maps 
the element  of our interest, \eqref{special},  to  the anomaly of the gravitino and the dilatino valued in $V$ together with the anomaly of spin-$\tfrac12$ fermions valued in $W$.
 This statement can be translated to mathematics most conveniently using the formulation of \cite{Yamashita:2021cao}. 
 
 To briefly recall it,  we have $(I_\bZ \Omega^\spin)^d(BG) \simeq (I_\bZ \Omega^{\spin \times G})^d(\pt)$, where $G$ is regarded as an internal symmetry group in the latter (i.e., the tangential structure given by the homomorphism $\Spin(d) \times G \xrightarrow{\mathrm{pr}_{\Spin (d)}} \Spin(d) \to \mathrm{O}(d)$). 
In \cite{Yamashita:2021cao}, a model of the differential extension of $(I_\bZ \Omega^{\spin\times G})^d(\pt)$ is given
 in terms of a pair $(\omega,h)$
 which directly formalizes the physics interpretation of invertible phases,
 where $\omega$ is called the anomaly polynomial and $h$ is the phase associated by the anomaly theory.
Mathematically,  $\omega$ is an element of total degree $d$ in 
$\H^\bullet (MT(\Spin\times G); \bR) \simeq   \H^\bullet(BG; \bR) \otimes_\bR \H^\bullet (B\Spin;\bR)$,
and can be identified with the rationalization of the  element in $(I_\bZ\Omega^{\spin \times G})^d(\pt)$.
Then, $h$ assigns a value in $\bR/\bZ$ to a $(d-1)$-dimensional closed spin manifold with connection $M$ which is equipped with a principal $G$-bundle with connection, 
so that when $M=\partial N$ they satisfy \begin{equation}
h(M) = \int_N \mathrm{cw}(\omega) \mod 1,
\end{equation}
where $\mathrm{cw}$ means the Chern-Weil construction. 
For example, $(\widehat{A}(T)|_{d}, \overline{\eta})$
is such a pair for $G = \{1\}$, as a result of the Atiyah-Patodi-Singer index theorem,
where $\widehat{A}(T)|_d \in {\H}^{d}(B\Spin; \bR)$ denotes the degree-$d$ part of the A-hat class of the universal bundle $T$ over $B \Spin$
and $\bar\eta:=(\eta(D)+\dim\ker D)/2$,  where $D$ is the Dirac operator of the spin bundle and $\eta$ is its eta invariant.

In our case, 
the image of an element of the form \eqref{special} is the pair $(\omega,h)$ given as follows,
by translating the physics result of \cite{AlvarezGaume:1983ig,Alvarez-Gaume:1984zst,Witten:1985xe} in the formulation we are using here.
We start with the case $d+22\equiv 0 \pmod 8$.
We can represent any element of $\KO^{d+22}_G \simeq \KO^0_G$ as a difference class of real representations of $G$.
Assume $V$ and $W$ are represented by real representations, and denote by $\mathrm{ch}(V), \mathrm{ch}(W) \in  \H^\bullet(BG; \bR)$ their Chern characters. 
Then $\omega$ is given by
\begin{equation}
\left.\omega = \frac12 \left(\mathrm{ch}(V) \otimes \widehat{A}(T) \cdot \mathrm{ch}(T)  - 4\mathrm{ch}(V) \otimes \widehat{A}(T) + \mathrm{ch}(W) \otimes \widehat{A}(T)
\right)\right|_{d+2},
\label{required}
\end{equation}
where $|_{d+2}$ means the degree $(d+2)$-part,
$\hat A$ is the A-hat polynomial,
and $T$ is the universal bundle over $B\Spin$ whose pullback under the classifying map $M\to B\Spin$ is the tangent bundle.

For a $(d+1)$-dimensional closed spin manifold $M$ with connection equipped with a principal $G$-bundle $P \to M$ with connection, $h$ is given by \begin{equation}
h(M)=\frac12[\bar\eta((P \times_G V) \otimes (TM\oplus \bR)) -4\bar\eta((P \times_G V))+ \bar\eta((P \times_G W)) ] \pmod \bZ,
\label{requiredEta}
\end{equation}
where $\bar\eta(E) \in \bR$ for a vector bundle with connection $E$ on a closed spin manifold
is defined as \begin{equation}
\bar\eta(E) := \frac12 (\eta(D_{S\otimes E}) + \dim\ker D_{S\otimes E} )
\end{equation} where $D_{S\otimes E}$ is the Dirac operator on the spin bundle  $S$ tensored with $E$.

In the case $22 + d \not\equiv 0 \pmod 8$, similar descriptions of the images of $\alpha_\spin$ in terms of twisted eta invariants are possible, 
since they can be represented by finite-dimensional Clifford modules with $G$-action \cite{AtiyahSegalCompletion}. 

\subsubsection{Anomalies of general elements in $\KO^\bullet_G((q))$}
Heterotic string constructions only produce elements of the form \eqref{special}
in $\KO^{22+d}_G((q))$, i.e.~those with poles of order at most 2.
To show that the heterotic string constructions do not have anomalies,
we have to show the vanishing only for this class of elements.
Therefore, we only have to define $\alpha_\spin$ on this class of elements.

We find it far more convenient, however, to define $\alpha_\spin$ on the entire elements of $\KO^{22+d}_G((q))$,
allowing poles of arbitrary order. 
This is because it allows us to use the power of modular forms and topological modular forms in the analysis.

For a general element  $U\in \KO^{22+d}_G((q))$, not necessarily of the form \eqref{special}, 
we generalize \eqref{required} and \eqref{requiredEta} as follows.
First, the expression \eqref{required} is generalized to
\begin{equation}
\omega = \text{coeff.~of $q^0$ of }  \frac12\left.\Delta(q) 
\ch(U) \otimes \widehat{A}(T) \cdot \mathrm{ch}(\Wit T)\right|_{d+2} ,
\label{omega}
\end{equation}
where we define\begin{equation}
\Wit V = \frac{\eta(q)^{d+2}}{q^{(d+2)/24}}  \bigotimes_{\ell\ge 1}
\bigoplus_{k\ge 0} q^{\ell k} \mathrm{Sym}^k V
\end{equation}
for a vector bundle $V$,
and $T$ is  the universal bundle over $B\Spin$ which pulls back to the tangent bundle by the classifying map of a spin manifold.
The element $\mathrm{ch}(\Wit T) \in \H^\bullet(B\Spin; \bR) \otimes_{\bR} \bR((q^{1/24}))$ is then the invariant polynomial which produces the characteristic form \begin{equation}
\ch(\Wit T) = \tr\left(
\frac{\eta(q)^{d+2}}{q^{(d+2)/24} \prod_\ell (1-q^\ell e^{iR/(2\pi)})}
\right),
\end{equation}
where $R$ is the curvature of the tangent bundle.
We then generalize the expression \eqref{requiredEta} to \begin{equation}
h(M) = \text{coeff.~of $q^0$ of }  
\frac12\Delta(q)  \bar\eta((P \times_G U)\otimes \Wit (TM\oplus \bR)).
\label{h}
\end{equation}
It is straightforward to check that the expressions \eqref{omega} and \eqref{h} reduce to \eqref{required} and \eqref{requiredEta} when restricted to elements of the form \eqref{special}.

\begin{rem}
\label{rem:wg}
This $\ch(\Wit T)$ is closely related to  the spectrum level expression of the Witten genus for spin manifolds which we denoted by $
\mathrm{Wit}_\spin : MT\Spin \to \KO((q)) 
$ in \eqref{slwg}.
Indeed, the Chern-Dold character of $\mathrm{Wit}_\spin$,
\begin{align}\label{eq_chd_Wit}
    \mathrm{chd}(\mathrm{Wit}_\spin ) \in \H^0(MT\Spin; \pi_{-\bullet}\mathrm{KO}((q)) \otimes \Q) \simeq H^\bullet(B\Spin; \Q) \otimes_\Q \Q((q)),
\end{align}
equals $\hat A(T)\cdot \ch(\Wit T)$. 
 Here, we use the homomorphism $\delta \colon\pi_{-\bullet}\KO((q)) \to \bZ((q))$ in \eqref{phiW} for the isomorphism in \eqref{eq_chd_Wit}. 
 In particular, for a spin manifold $N$ and its class $[N]\in \Omega^\spin_\bullet (\pt)$, we have \begin{equation}
\delta\circ \Wit_\spin([N]) =\int_N \widehat{A}(TN) \cdot \ch(\Wit TN) = \int_N \chd(\Wit_\spin).
\end{equation}
\end{rem}

\subsubsection{Generalization from $BG$ to more general $X$}

So far we considered truly equivariant versions of $\TMF$ and $\KO$, 
but we expect that the maps factor through the Borel equivariant versions.
Then there is no reason to restrict the arguments to the various generalized cohomology theories to be the classifying spaces,
and we will replace $BG$ by a more general $X$, and require the naturality for $\alpha_{\mathrm{string}}$ and $\alpha_{\mathrm{spin}}$. 
This move to $X$ has an added bonus that it can also take into account the possible existence of the space of exactly marginal couplings of the input 2d superconformal field theory $T$, 
which appears as the target space of massless scalar fields of the spacetime theory,
which can also have an anomaly \cite{Moore:1984dc,Moore:1984ws,Manohar:1984zj}.
Summarizing, we expect that there exists a natural transformation \begin{equation}
\alpha_\spin: \KO^{22+d\oplus -\tilde k}((q))(X) \to (I_\bZ\Omega^\spin)^{d+2\oplus -\tilde k}(X).
\end{equation}

We can describe this transformation in the case $22+d \equiv 0 \pmod 4$ in terms of twisted Dirac operators as in the previous cases. 
However, in the other degrees, this does not generalize straightforwardly
since it is not true in general that we can represent elements of $\KO$-groups as difference classes of finite-dimensional Clifford module bundles \cite{KaroubiCliff} unless $22+d\equiv 0 \pmod 4$ (e.g., consider the generator of $\KO^1(S^1) \simeq \bZ$).
One possible way to describe the image of general elements is to use suspension to reduce to the case $d + 22 \equiv 0 \pmod 4$, assuming that $\alpha_\spin$ is a transformation of cohomology theories. 
We can check that, in the case of $\KO^{22+d}_G((q))(\pt)$, the two constructions produce the same elements.

\subsubsection{Compatibility with compactifications}

There is another property of $\alpha_\stri$ and $\alpha_\spin$ we expect from physics considerations and assume in the following.
Given a $d$-dimensional heterotic compactification given by a class $[T]\in \TMF^{22+d}(X)$,
we can make a further compactification on an $m$-dimensional smooth string manifold $M$,
resulting in a $(d-m)$-dimensional compactification.
The anomaly of this $(d-m)$-dimensional compactification can be computed in two ways.
One is first to  compute the anomaly in $d$ dimensions in $(I_\bZ\Omega^\stri)^{d+2}(X)$, 
and then to compactify the resulting the anomaly on $M$, which is valued in $(I_\bZ\Omega^\stri)^{d-m+2}(X)$.
Another is first to  consider the internal superconformal field theory describing $[T]$ together with the motion along the manifold $M$ in $\TMF^{22+d-m}(X)$,
and then compute its anomaly via $\alpha_\stri$, again resulting in $(I_\bZ\Omega^\stri)^{d-m+2}(X)$.
We expect and assume that these two give the same result. In other words we assume the commutativity of the following square
    \begin{align}
        \vcenter{\xymatrix{
        \TMF^{d+22} \ar[r]^-{\alpha_{\stri}}(X)\ar[d]^-{\mathrm{Wit}_\stri([M])\cdot } & (I_\bZ\Omega^{\stri})^{d+2} (X)\ar[d]^-{[M] \cdot} \\
        \TMF^{d-m+22}(X) \ar[r]^-{\alpha_{\stri}}& (I_\bZ\Omega^{\stri})^{d-m+2} (X)
        }}.
    \end{align}
The anomaly from fermions are expected to behave in a similar manner, so we assume the following commuting square:
    \begin{align}
        \vcenter{\xymatrix{
        \KO((q))^{d+22}(X) \ar[r]^-{\alpha_{\spin}}\ar[d]^-{\mathrm{Wit}_\spin([M])\cdot } & (I_\bZ\Omega^{\spin})^{d+2}(X) \ar[d]^-{[M] \cdot} \\
        \KO((q))^{d-m+22}(X) \ar[r]^-{\alpha_{\spin}}& (I_\bZ\Omega^{\spin})^{d-m+2} (X)
        }}.
    \end{align}
\noindent The expression we gave for $\omega$ in \eqref{omega} is compatible with this requirement, thanks to Remark~\ref{rem:wg} and the product formula for the Witten genus.

More generally, we can consider a family of $m$-dimensional string manifolds $M$ parameterized by $X$, described by a fibration $M \to N\to X$.
We can then consider a family of superconformal field theories parameterized by $N$, specified by a class in $\TMF^{d+22}(N)$.
The corresponding anomaly takes values in $(I_\bZ \Omega^\stri)^{d+2}(N)$.
Let us now regard the fiber $M$ as a part of the spacetime and compactify along it.
We now have a family of theories parameterized by $X$, specified by a class in $\TMF^{d-m+22}(X)$,
whose anomaly takes values in $(I_\bZ \Omega^\stri)^{d-m+2}(X)$.
We then expect and assume that two ways of computing anomalies are equal, i.e.~we demand the commutativity of the following square     \begin{align}
        \vcenter{\xymatrix{
        \TMF^{d+22} \ar[r]^-{\alpha_{\stri}}(N)\ar[d]^-{\text{pushforward}} & (I_\bZ\Omega^{\stri})^{d+2} (N)\ar[d]^-{\text{pushforward}} \\
        \TMF^{d-m+22}(X) \ar[r]^-{\alpha_{\stri}}& (I_\bZ\Omega^{\stri})^{d-m+2} (X)
        }}
    \end{align}
and similarly 
\begin{align}
        \vcenter{\xymatrix{
        \KO((q))^{d+22} \ar[r]^-{\alpha_{\spin}}(N)\ar[d]^-{\text{pushforward}} & (I_\bZ\Omega^{\spin})^{d+2} (N)\ar[d]^-{\text{pushforward}} \\
        \KO((q))^{d-m+22}(X) \ar[r]^-{\alpha_{\spin}}& (I_\bZ\Omega^{\spin})^{d-m+2} (X)
        }}.
\end{align}

All these properties follow most naturally if the spectrum level expressions of $\alpha_\stri$ and $\alpha_\spin$, namely the morphisms \begin{align}
\alpha_\stri\colon&  \TMF \to \Sigma^{-20}I_\bZ MT\Stri,  &
\alpha_\spin\colon&  \KO((q)) \to \Sigma^{-20}I_\bZ MT\Spin,  
\end{align}
preserve $MT\Stri$- and $MT\Spin$-module structures, 
where the module structures of $\TMF$ and $\KO((q))$ are given in terms of $\Wit_\stri$ and $\Wit_\spin$ in \eqref{slwg}, respectively.
We use them as one of the assumptions in Sec.~\ref{sec:proof}.

\subsubsection{Summary}\label{subsubsec_physfacts}
We summarize the long discussion in this section as three physics assumptions used as the starting point in the rigorous proof given in the next section.
Namely,
the  natural transformation of our interest is given as follows:
\begin{physfact}
\label{physicsfact}
There exists a morphism of spectra,  
\[
    \alpha_{\spin} \colon \KO((q)) \to \Sigma^{-20}I_\bZ MT\mathrm{Spin},
\]
so that the fermion anomaly of the heterotic compactification to $d$ dimensions with a path-connected parameter space $X$ with level $k\in (\widetilde{I_\bZ\Omega^\spin})^4(X)$ 
is characterized by a natural transformation obtained by the composition \begin{multline*}
\alpha_\stri:\TMF^{22+d\oplus -k}(X) \xrightarrow{\sigma} \KO^{22+d\oplus -\tilde k}((q))(X)\\
\xrightarrow{\alpha_\spin} (I_\bZ\Omega^\spin)^{d+2\oplus -\tilde k}(X)
\xrightarrow{I_\bZ \iota} (I_\bZ\Omega^\text{string})^{d+2\oplus -k}(X),
\end{multline*}
where $\sigma$ is the natural transformation in \eqref{slwg}
and $I_\bZ \iota$ is the Anderson dual of the forgetful map $\iota \colon MT\mathrm{String} \to MT\mathrm{Spin}$ which also appeared in \eqref{slwg}.
\end{physfact}
\begin{rem}
Physically, this statement means that we compute the anomaly of a heterotic compactification by 
first extracting the massless fermion fields by $\sigma$,
computing their anomaly by $\alpha_\spin$,
and restricting the spacetime manifolds to satisfy the Green-Schwarz constraint ``$dH=\tr R^2-k\,\tr F^2$'' by $I_\bZ\iota$.
\end{rem}

For the  properties of $\alpha_\spin$, we will  only use the following two:
\begin{physfact}
\label{fact_rationalization}
The rationalization of $\alpha_\spin$, \[
(\alpha_\spin)_\Q: \KO((q))^{d + 22}_\Q(\pt)  \to (I_\bZ\Omega^\spin)^{d + 2}_\Q(\pt) \simeq \mathrm{H}^{d+2}(B\mathrm{Spin}; \Q)
\] 
is given by
\begin{align*}
    U \mapsto \text{coeff.~of $q^0$ of } \left( \tfrac12 \Delta(q) \chd(\Wit_\spin ) \cdot \ch(U)\right). 
\end{align*}
Here we are using the homomorphism $\delta \colon\pi_{-\bullet}\KO((q)) \to \bZ((q))$ in \eqref{phiW}. 
\end{physfact}
\begin{rem}
Physically, this simply describes the anomaly polynomial of the spacetime theory in terms of the elliptic genus of the internal SCFT, as was done already in the 80s in \cite{Schellekens:1986xh,Lerche:1987qk,Lerche:1988np}.
\end{rem}
 \begin{physfact}
\label{fact_natural}
The morphism $\alpha_\spin$ preserves the $MT\mathrm{Spin}$-module structures, where the $MT \mathrm{Spin}$-module structure of $\KO((q))$ is given by $\mathrm{Wit}_\spin$ in \eqref{slwg}. 
\end{physfact}
\begin{rem}
Physically, this means that the anomaly of the theory obtained by a further compactification on a smooth manifold $M$ 
can also be computed by first computing the anomaly in higher dimensions and 
then later evaluating that on the manifold $M$.
\end{rem}

\section{Vanishing of anomalies}
\label{sec:proof}
\subsection{The proof}
In the previous sections, we have translated our physics problem into a purely mathematical one,
the vanishing of $\alpha_\stri$.
Here we give  a mathematical proof to it,
starting from the physical assumptions listed in Subsubsection \ref{subsubsec_physfacts}.
The result is our main Theorem \ref{thm_vanishing}.

We work in the stable homotopy category, and start from the following general result.
\begin{lem}\label{lem_multiplicative}
Let $E$ be a multiplicative cohomology theory and $\mathcal{B}$ be a multiplicative tangential structure. 
Assume we are given a homomorphism of ring spectra, 
\begin{align}
    \mathcal{G} \colon MT\mathcal{B} \to E. 
\end{align}
Let $n$ be an integer and assume that a morphism
\begin{align}\label{eq_alpha}
    \alpha \colon E \to  \Sigma^n I_\bZ  MT\mathcal{B}
\end{align}
preserves the $MT\mathcal{B}$-module structures. 
Then there exists a unique element $\beta \in I_\bZ E^n(\pt)$ such that $\alpha \in [E, \Sigma^n I_\bZ MT{\mathcal{B}} ] \simeq [ MT\mathcal{B} \wedge E , \Sigma^n\IZ]$ coincides with the following composition, 
\begin{align}\label{eq_lem_multiplicative}
     MT\mathcal{B} \wedge E \xrightarrow{ \mathcal{G} \wedge \mathrm{id}} E \wedge E \xrightarrow{\mathrm{multi}} E \xrightarrow{\beta} \Sigma^n \IZ. 
\end{align}

\end{lem}
\begin{proof}
The Anderson dual to \eqref{eq_alpha} is denoted by $I_\bZ \alpha \in [MT\mathcal{B}, \Sigma^n I_\bZ E]$. 
Denote the unit of $MT{\mathcal{B}}$ by $1_{MT{\mathcal{B}}} \in \pi_0(MT{\mathcal{B}})$. 
We define the element $\beta \in I_\bZ E^n(\pt)$ by
\begin{align}
    \beta := I_\bZ \alpha (1_{MT{\mathcal{B}}}). 
\end{align}
The assumption that $\alpha$ is an $MT \mathcal{B}$-module homomorphism implies that $I_\bZ \alpha$ is also an $MT \mathcal{B}$-module homomorphism, so it is given by \eqref{eq_lem_multiplicative} under $[MT\mathcal{B}, \Sigma^n I_\bZ E] \simeq [MT\mathcal{B}\wedge E , \Sigma^n\IZ]$. 
This is equivalent to the statement of the Lemma. 
\end{proof}

\begin{rem}\label{rem_ABS}
An example of transformations of the form \eqref{eq_alpha} appears in \cite[Section 9]{FreedHopkins2021}, namely the case $s = 0$ in their notation;  
see also \eqref{freefermion}. 
As conjectured in \cite[Conjecture 9.70]{FreedHopkins2021}, the natural transformation they describe is expected to coincide with \eqref{eq_alpha} where $\mathcal{G} = \mathrm{ABS} \colon MT\Spin \to KO$ is the Atiyah-Bott-Shapiro orientation and $\beta = \gamma_{\mathrm{KO}} \in I_\bZ \KO^{4}(\pt)$ is the Anderson self-duality element for $\KO$. 
The same transformation also appears in \eqref{eq_proof_twist_TMF}. 
\end{rem}

Using Lemma \ref{lem_multiplicative}, we now identify the morphism $\alpha_\spin$. 
We have the following canonical identification, 
\begin{align}
     I_\bZ \mathrm{KO}((q))^{-20}(\pt) \simeq \mathrm{Hom}(\pi_{-20}\mathrm{KO}((q)), \bZ) \simeq \mathrm{Hom}(\bZ((q)), \bZ), 
\end{align}
where the first isomorphism follows by $\pi_{-21}\mathrm{KO}((q)) = 0$ and \eqref{eq_exact_IE}, and the second isomorphism uses the generator which maps to $2$ by $\delta$ in \eqref{phiW}. 
We have the element $\Delta(q) \cdot - |_{q^0} \in \mathrm{Hom}(\bZ((q)), \bZ)$, so we denote the corresponding element in $I_\bZ \mathrm{KO}((q))^{-20}(\pt) $ by the same symbol. 
\begin{lem}\label{ass_simple}
The morphism $\alpha_\spin \in [\mathrm{KO}((q)), \Sigma^{-20}I_\bZ MT \Spin] = [ MT \Spin \wedge \mathrm{KO}((q)) , \Sigma^{-20}\IZ]$ coincides with the following composition, 
\begin{align}\label{eq_ass_simple}
MT \Spin \wedge \mathrm{KO}((q)) \xrightarrow{ \mathrm{Wit}_\spin \wedge \mathrm{id}} 
    \mathrm{KO}((q)) \wedge \mathrm{KO}((q))\xrightarrow{\mathrm{multi}} \mathrm{KO}((q)) \xrightarrow{\Delta(q) \cdot - |_{q^0}} \Sigma^{-20} \IZ.
\end{align}
\end{lem}
\begin{proof}
The morphism $\alpha_\spin$ preserves the $MT \Spin$-module structure as stated in \PhysicsFact~ \ref{fact_natural}.
Then,  Lemma \ref{lem_multiplicative} implies that there exists a unique element $\beta_\spin \in I_\bZ \mathrm{KO}((q))^{-20}(\pt)$ such that $\alpha_\spin$ coincides with the composition
\begin{align}
MT \Spin \wedge \mathrm{KO}((q)) \xrightarrow{ \mathrm{Wit}_\spin \wedge \mathrm{id}} 
    \mathrm{KO}((q)) \wedge \mathrm{KO}((q))\xrightarrow{\mathrm{multi}} \mathrm{KO}((q)) \xrightarrow{\beta_\spin} \Sigma^{-20} \IZ.
\end{align}
It is enough to show $\beta_\spin = \Delta(q) \cdot - |_{q^0}$. 
By \PhysicsFact~ \ref{fact_rationalization}, we know that these elements are equal after rationalization, i.e., the images in $[\mathrm{KO}((q))_\Q, \Sigma^{-20}I\Q]$ coincide (note that $\frac{1}{2}$ does not appear here because $\delta$ multiplies the generator of $\pi_{-20}\mathrm{KO}$ by two). 
But since $\pi_{-21}(\mathrm{KO}((q)))= 0$ the rationalization homomorphism
\begin{align}
    [\mathrm{KO}((q)),\Sigma^{-20}\IZ] \to [\mathrm{KO}((q))_\Q, \Sigma^{-20}I\Q]
\end{align}
is injective by the exactness of \eqref{eq_exact_IE}, so we get the result. 
\end{proof}

Now we proceed to the proof of the vanishing of $\alpha_\stri$.
We start with a few mathematical facts:

\begin{fact}\label{fact_TMF_21}
We have
\[
    \mathrm{TMF}^{21}(\pt) = 0. 
\]
\end{fact}
\begin{proof}
See the Table in \cite[Chapter 13]{TMFBook},
which is reproduced in Appendix~\ref{sec:tmfdata}.
\end{proof}

\begin{fact}
\label{fact:mf}
The ring $\MF$  of integral modular forms has the $\bZ$-basis
given by $c_4^i c_6^j \Delta^k$ where $i\ge 0$; $j=0,1$; $k\ge 0$.
Here,  $c_4=1+240\sum_{n\ge 1} \sigma_3(n) q^n$ and $c_6=1-504\sum_{n\ge 1} \sigma_5(n) q^n$ are the Eisenstein series of degree 4 and 6 normalized to have integer $q$-expansion coefficients,
and $\Delta$ is the modular discriminant and satisfies $1728\Delta=c_4^3-c_6^2$.
\end{fact}
\begin{proof}
Integral modular forms can be defined in two ways,
one as modular forms associated to elliptic curves over $\bZ$,
and another as modular forms over $\bC$ whose $q$-expansion coefficients are in $\bZ$.
These two definitions give the same ring, a posteriori.
The ring  of integral modular forms in the former sense was determined in \cite{Deligne} and was shown to have the form given in the theorem.
The ring of integral modular forms in the latter sense can be determined by first noticing that 
$c_4^i c_6^j$ for $4i+6j< 12$ generate the modular forms of weight less than 12 over $\bC$.
We can then prove the statement by induction: given a modular form $f$ of some weight $k \ge 12$ with integral $q$-expansion coefficients,
we  consider $g=f- f_0 c_4^i c_6^j$, where $4i+6j=k$ and $f_0$ is the constant term of $f$.
This $g$ is a cusp form of weight $k$, and therefore $g/\Delta$ is a modular form of weight $k-12$ with integral $q$-expansion coefficients, proving the Fact.
\end{proof}

\begin{lem}
\label{lem:q}
The constant term of the $q$-expansion of any weakly-holomorphic modular form of degree two vanishes.
\end{lem}
\begin{proof}
From Fact~\ref{fact:mf},
any weakly-holomorphic modular form of degree 2 is a linear combination of 
\begin{equation}
c_4^{-1+3n} c_6  \Delta^{-n} = (c_6 /c_4) j^n 
\end{equation} for $n\ge 1$,
where $j$ is the modular $j$-function, $j=c_4^3/\Delta$.
It turns out that $q\tfrac{d}{dq}j=-j c_6/c_4$, and therefore $(c_6/c_4) j^n = -\tfrac1n q\tfrac{d}{dq} j^n$.
Therefore the constant term in the $q$-expansion vanishes.
\end{proof}

With these facts, we can finish our proof.
We first show that the free part of $\alpha_\stri$ vanishes by working over $\Q$:

\begin{lem}\label{fact_rational_vanishing}
Let $\beta_\stri: \TMF\to \Sigma^{-20}I\bZ$ 
be the composition of $\sigma :\TMF\to \KO((q))$ and $\beta_\spin: \KO((q))  \to \Sigma^{-20}I\bZ$.
This vanishes rationally.
\end{lem}
\begin{proof}
Rationally, we can replace $\KO^{\bullet}((q))(\pt)$ by $\MF[\Delta^{-1}]_{-\bullet/2}$.
It then suffices to show that the coefficient of $q^0$ of $\Delta(q) \phi(x)$ is zero for all $x\in \TMF^{20}(\pt)$.
Since $\Delta(q)\phi(x)$ is a weakly-holomorphic modular form of degree $(-20+24)/2 = 2$,
this follows from Lemma \ref{lem:q}.
\end{proof}
We note that this reduction to the vanishing of the constant term of  weakly-holomorphic modular forms of degree 2 was essentially how the vanishing of the perturbative anomalies of general heterotic compactifications was shown in \cite{Schellekens:1986xh,Lerche:1987qk,Lerche:1988np}.

Our remaining task is to show that the torsion part also vanishes:
\begin{thm}\label{thm_vanishing}
The composition
\[
    \alpha_\stri \colon \mathrm{TMF} \xrightarrow{\sigma}
    \KO((q)) \xrightarrow{\alpha_{\spin}} \Sigma^{-20} I_\bZ MT{\Spin} \xrightarrow{I_\bZ\iota}
   \Sigma^{-20}  I_\bZ MT{\mathrm{String}}
\]
is zero. 
\end{thm}
\begin{proof}
Lemma \ref{ass_simple} means that $\alpha_\stri \in [\mathrm{TMF},\Sigma^{-20}I_\bZ MT{\mathrm{String}} ] = [ MT{\mathrm{String} \wedge \TMF}, \Sigma^{-20}\IZ]$ is given by the composition 
\begin{align}
    MT\mathrm{String} \wedge \TMF \xrightarrow{ \mathrm{Wit}_\stri \wedge \mathrm{id}} \mathrm{TMF}  \wedge \mathrm{TMF}  \xrightarrow{\mathrm{multi}} \mathrm{TMF}  \xrightarrow{\beta_\stri} \Sigma^{-20}\IZ, 
\end{align}
where $\beta_\stri$ is the image of $\Delta(q) \cdot - |_{q^0} \in I_\bZ\KO((q))^{-20}(\pt)$ under the Anderson dual to the canonical map $\mathrm{TMF} \to \KO((q))$. 
By Lemma \ref{fact_rational_vanishing} we know that the rationalization of $\beta_\stri$ in $[\mathrm{TMF}_\Q ,\Sigma^{-20}I\Q]$ is zero. 
By Fact \ref{fact_TMF_21} and the exactness of \eqref{eq_exact_IE}, we see that the rationalization
\begin{align}
    [\mathrm{TMF}, \Sigma^{-20}\IZ] \to [\mathrm{TMF}_\Q ,\Sigma^{-20}I\Q]
\end{align}
is injective, so we get $\beta_\stri=0$ and the result follows. 
\end{proof}

\begin{rem}
We remind the reader that this main theorem establishes that there is no  anomalies in arbitrary perturbative compactification of heterotic string theory.
\end{rem}

\subsection{A corollary}
\label{sec:implication}
As a last statement in the main part of the paper, we prove the following corollary of our main theorem \ref{thm_vanishing}, by considering the 
particular case of $d=2$ and $X=\pt$, which was originally discussed 
in \cite{Tachikawa:2021mvw} in a physics language.
We present this result here, since the appearance of the pairing with the combined degree $-24+3=-21$ is reminiscent of the Anderson self-duality of $\Tmf$, 
which is very briefly reviewed in Appendix.~\ref{sec:tmfdata}.
 \begin{cor}
 \label{cor:cor}
There is a natural perfect pairing between the cokernel $\bZ/24\bZ$ of \begin{equation}
\phi: \TMF_{-24}(\pt) \to \MF[\Delta^{-1}]_{-12} 
\end{equation} and the kernel $\bZ/24\bZ$ of \begin{equation}
\phi: \TMF_{3}(\pt) \to \MF[\Delta^{-1}]_{3/2} .
\end{equation}
\end{cor}

The proof of this corollary is based on the following fact:
\begin{fact}[=\cite{Hopkins2002}, Proposition 4.6]
\label{fact:image}
The image of $\phi$ has a $\bZ$-basis given by \begin{equation}
a_{i,j,k} c_4^i c_6^j \Delta^k, \qquad \text{$i\ge 0$; $j=0,1$; $k\in\bZ$}
\end{equation}where \begin{equation}
a_{i,j,k} = \begin{cases}
24/\gcd(24,k) & \text{if $i=j=0$},\\
2 & \text{if $j=1$},\\
1 & \text{otherwise}.
\end{cases}
\end{equation}
\end{fact}

\begin{proof}[Proof of Corollary~\ref{cor:cor}]
 Our Theorem \ref{thm_vanishing} implies the vanishing of the composition \begin{equation}
\alpha_\stri: \TMF^{24}(\pt) \xrightarrow{\sigma} \underbrace{\KO^{24}((q))(\pt) }_{\bZ((q))}
\xrightarrow{\alpha_\spin} \underbrace{(I_\bZ\Omega^\spin)^4(\pt) }_\bZ
\xrightarrow{I_\bZ \iota} \underbrace{(I_\bZ\Omega^\stri)^4(\pt)}_{\bZ/24\bZ}.
\end{equation}
For this specific degree, the commutative diagram \eqref{phiW} shows that $\sigma$
factors as \begin{equation}
\sigma: \TMF^{24}(\pt) \xrightarrow{\phi} \MF[\Delta^{-1}]_{-12} 
\xrightarrow{\text{$q$-exp.}} \bZ((q))
\stackrel{\delta^{-1}}{\simeq} \KO^{24}((q))(\pt).
\end{equation}
We then slightly rewrite $\alpha_\stri$ as follows: \begin{equation}
\TMF^{24}(\pt) \xrightarrow{\phi} \MF[\Delta^{-1}]_{-12} 
\xrightarrow{f} \underbrace{(I_\bZ\Omega^\spin)^4(\pt) }_\bZ
\xrightarrow{I_\bZ \iota} \underbrace{(I_\bZ\Omega^\stri)^4(\pt)}_{\bZ/24\bZ}.
\end{equation}
where $f=\alpha_\spin \circ \delta^{-1} \circ \text{$q$-expansion}$.
As the composition vanishes, we get a homomorphism $f' \colon \mathrm{Coker} \phi  \to {(I_\bZ\Omega^\stri)^4(\pt)}$ so that the following diagram commutes. 
\begin{align}\label{foo}
    \xymatrix{
    \TMF^{24}(\pt) \ar[r]^-{\phi} &  \MF[\Delta^{-1}]_{-12} \ar[r]^-{f} \ar[d]
& {(I_\bZ\Omega^\spin)^4(\pt) }\ar[r]^-{I_\bZ \iota}
& {(I_\bZ\Omega^\stri)^4(\pt)} \\
    & \mathrm{Coker} \phi \ar[rru]_-{f'}&&
    }
\end{align}

Now, the cokernel of $\phi$ in our case can be found via Fact~\ref{fact:image} to be $\bZ/24\bZ$ generated by $\Delta^{-1}$.
An explicit computation shows that $f(\Delta^{-1})$ is a generator of $(I_\bZ\Omega^\spin)^4(\pt)\simeq \bZ$.
We also show in Corollary \ref{cor_string_spin_dual} that $I_\bZ \iota$
maps a generator to a generator.
Therefore, the map $f'$ in \eqref{foo} is actually an isomorphism.
Then the result follows from the fact that $(I_\bZ\Omega^\stri)^4(\pt)$ is the Pontrjagin dual to $\Omega^\stri_3(\pt)\simeq \TMF_3(\pt)\simeq \bZ/24\bZ$.
\end{proof}

\appendix

\section{Table of spin and string bordism groups}
\label{sec:bordismdata}

For the convenience of the readers,
we provide the table of spin and string bordism groups,
taken from \cite{ABP,Giambalvo1971},
in Table~\ref{table:bordism}.
We use  the abbreviations $\bZ_a:=\bZ/a\bZ$.
We note that the string bordism groups are equal to the framed bordism up to $d=6$,
and to $\pi_d(\tmf)$ up to $d=14$.

\begin{table}[h]
\[
\begin{array}{c|ccccccccccccccccccccc}
d& 0 & 1 & 2 & 3 & 4 & 5 & 6 & 7 & 8 & 9 & 10 & 11 & 12  & 13 & 14 & 15 & 16
\\
\hline
\Omega^\spin_d(\pt) &
\bZ & \bZ_2 & \bZ_2 & 0 & \bZ & 0& 0 & 0 & \bZ^2  & (\bZ_2)^2 & (\bZ_2)^3 & 0 & \bZ^3 & 0 & 0 &0 & \bZ^5
\\
\Omega^\text{string}_d(\pt) &
\bZ &
\bZ_2 &
\bZ_2 &
\bZ_{24} &
0 &
0 &
\bZ_2 & 
0 &
\bZ\oplus \bZ_2 &
(\bZ_2)^2 &
\bZ_6&
0 &
\bZ &
\bZ_3 &
\bZ_2 &
\bZ_2 &
\bZ^2
\end{array}
\]
\caption{Table of spin and  string bordism groups\label{table:bordism}}
\end{table}

\section{Tables of $\pi_*(\tmf)$, $\pi_*(\Tmf)$ and $\pi_*(\TMF)$}
\label{sec:tmfdata}

Here we reproduce the table of $\tmf_\nu(\pt)=\pi_\nu(\tmf)$ from \cite[Chap.~13]{TMFBook}
in Table~\ref{table:2} and Table~\ref{table:3} for the convenience of the readers;
the authors think that the 576-periodic homotopy groups of $\TMF$ should be as well-known as the 8-periodic homotopy groups of $\KO$.
The table should be used in the following manner. 
We first note that there is a morphism  $\phi:\pi_d(\tmf)\to \MF_{d/2}$, 
where $\MF_*=\bZ[c_4,c_6,\Delta]/(c_4^3-c_6^2-1728\Delta)$
where $c_4$, $c_6$ and $\Delta$ has degree $2$, $3$ and $6$, respectively.
Its cokernel is given by Fact~\ref{fact:image}.
The kernel of $\phi$ consists of torsion elements of $\pi_*(\tmf)$, whose order is of the form $2^a 3^b$.
Therefore, what remains to be known to determine $\pi_*(\tmf)$ is the data\footnote{%
This is a footnote for physicists unfamiliar with the notations.
For an Abelian group $A$, $A_{(p)}$ denotes its localization at $p$, 
i.e.~an abelian group obtained by adjoining inverses of primes other than $p$.
For primes $p$ and $q$, $(\bZ/p^n\bZ)_{(q)} $  is $\bZ/p^n\bZ$ if $p=q$ and is $0$ if $p\neq q$.
This allows one to reconstruct any finitely generated Abelian group $A$ from $A_{(p)}$ for all $p$.
For example, from the tables, $\pi_{20}(\tmf)_{(2)} = \bZ_{(2)}\oplus \bZ/8\bZ$ and 
$\pi_{20}(\tmf)_{(3)} = \bZ_{(3)}\oplus \bZ/3\bZ$.
This means that $\pi_{20}(\tmf)=\bZ \oplus \bZ/24\bZ$.
} of $\pi_*(\tmf)_{(2)}$ and $\pi_*(\tmf)_{(3)}$, 
which are provided in Table~\ref{table:2} and Table~\ref{table:3}, respectively.
There, the abbreviations $\bZ_a:=\bZ/a\bZ$ are used.
Each entry of Table~\ref{table:2} and Table~\ref{table:3} are separated into the first row and the second row:
\begin{itemize}
\item For $n$ divisible by eight, the first row contains elements which are pre-images of $c_4^a c_6^b\Delta^c$.
\item For $n\equiv 4$ modulo 8, the first row contains elements which are pre-images of $2c_4^a c_6^b\Delta^c$.
\item For $n\equiv 1$ or $2$ modulo 8, the first row contains elements which are obtained by multiplying $\eta$ or $\eta^2$  to the pre-images of $c_4^a c_6^b \Delta^c$,
where $\eta\in \pi_1(\tmf)$ is the class defined by $S^1$ with periodic spin structure.
\item The second row contains other elements, which are mostly torsion except the piece $\bZ_{(2)}$ or $\bZ_{(3)}$ generated by the pre-image of $\tfrac{24}{\gcd(24,k)}\Delta^k$.
\item The total number of $\bZ_{(2)}$ or $\bZ_{(3)}$ is equal to the number of solutions to $4a+6b+12c=n$.
\item The second row is 192-periodic for Table~\ref{table:2} 
and is 72-periodic for Table~\ref{table:3}.
\end{itemize}

$\pi_*(\TMF)$ is obtained by inverting (the pre-image of) $\Delta^{24}$ of $\pi_*(\tmf)$.
More concretely, it is obtained by replacing the first row by
\begin{equation}
\bZ_{(2)}[x],\bZ_{2}[x],\bZ_{2}[x],0,\bZ_{(2)}[x],0,0,0
\end{equation}
for Table~\ref{table:2} and \begin{equation}
\bZ_{(3)}[x],0,0,0,\bZ_{(3)}[x],0,0,0
\end{equation}
for Table~\ref{table:3}.

\begin{table}
\[
\hskip-1cm\begin{array}{c|ccccccccccccccccccccc}
\hline\hline
d &  0 & 1 & 2 & 3 & 4 & 5 & 6 & 7 & 8 & 9 & 10 & 11 & 12  & 13 & 14 & 15  \\
\hline
\pi_d(\tmf)_{(2)}& \bZ_{(2)} & \bZ_2 & \bZ_2 &  &  && &  & \bZ_{(2)}  & \bZ_2 & \bZ_2 &  & \bZ_{(2)} &  &  &  \\
&  &  &  & \bZ_8 &  &&\bZ_2 &  &  \bZ_2 & \bZ_2 &  &  &  & & \bZ_2 & \bZ_2 \\
\hline\hline
d &  16 & 17 & 18 & 19 & 20 & 21 & 22 & 23 & 24 & 25 & 26 & 27 & 28  & 29 & 30 & 31  \\
\hline
\pi_d(\tmf)_{(2)}& \bZ_{(2)} & \bZ_2 & \bZ_2 &  & \bZ_{(2)}  && &  & \bZ_{(2)}  & \bZ_2 & \bZ_2 &  & \bZ_{(2)} &  &  &  \\
&  & \bZ_2 &  &  &   \bZ_8&\bZ_2& \bZ_2 &  & \bZ_{(2)}  &\bZ_2  & \bZ_2  &  \bZ_4 &  \bZ_2 & & &  \\
\hline\hline
d &  32&33&34&35&36&37&38&39&40&41&42&43&44&45&46&47\\
\hline
\pi_d(\tmf)_{(2)}& \bZ_{(2)}^2 & \bZ_2^2 & \bZ_2^2 &  & \bZ_{(2)}^2  && &  & \bZ_{(2)}^2  & \bZ_2^2 & \bZ_2^2 &  & \bZ_{(2)}^2 &  &  &  \\
&  \bZ_2& \bZ_2 & \bZ_2 & \bZ_2 &   &&  &\bZ_2  & \bZ_4  &\bZ_2  & \bZ_2  &  &  &\bZ_2 &\bZ_2 &  \\
\hline\hline
d & 48&49&50&51&52&53&54&55&56&57&58&59&60&61&62&63\\
\hline
\pi_d(\tmf)_{(2)}& \bZ_{(2)}^2 & \bZ_2^2 & \bZ_2^2 &  & \bZ_{(2)}^2  && &  & \bZ_{(2)}^3  & \bZ_2^3 & \bZ_2^3 &  & \bZ_{(2)}^3 &  &  &  \\
&\bZ_{(2)} &&  \bZ_2 & \bZ_8 & \bZ_2 &\bZ_2 &  \bZ_4 & &&\bZ_2&&\bZ_2&\bZ_4\\
\hline\hline
d & 64&65&66&67&68&69&70&71&72&73&74&75&76&77&78&79\\
\hline
\pi_d(\tmf)_{(2)}^3& \bZ_{(2)}^3 & \bZ_2^3 & \bZ_2^3 &  & \bZ_{(2)}^3  && &  & \bZ_{(2)}^3  & \bZ_2^3 & \bZ_2^3 &  & \bZ_{(2)}^3 &  &  &  \\
& & \bZ_2^2 & \bZ_2 & & \bZ_2 & & \bZ_2 & & \bZ_{(2)} & & & \bZ_2 & & &  &\\
\hline\hline
d &80&81&82&83&84&85&86&87&88&89&90&91&92&93&94&95\\
\hline
\pi_d(\tmf)_{(2)}& \bZ_{(2)}^4 & \bZ_2^4 & \bZ_2^4 &  & \bZ_{(2)}^4  && &  & \bZ_{(2)}^4  & \bZ_2^4 & \bZ_2^4 &  & \bZ_{(2)}^4 &  &  &  \\
& \bZ_2 & & & & & \bZ_2 & & & & & \bZ_2 & & & & &\\
\hline\hline
d & 96&97&98&99&100&101&102&103&104&105&106&107&108&109&110&111\\
\hline
\pi_d(\tmf)_{(2)}& \bZ_{(2)}^4 & \bZ_2^4 & \bZ_2^4 &  & \bZ_{(2)}^5  && &  & \bZ_{(2)}^5  & \bZ_2^5 & \bZ_2^5 &  & \bZ_{(2)}^5 &  &  &  \\
& \bZ_{(2)}& \bZ_2 &\bZ_2 & \bZ_8 & \bZ_2 & & \bZ_2 & & \bZ_2& \bZ_2^2 & & & & & \bZ_4 & \bZ_2\\
\hline\hline
d & 112&113&114&115&116&117&118&119&120&121&122&123&124&125&126&127\\
\hline
\pi_d(\tmf)_{(2)}& \bZ_{(2)}^5 & \bZ_2^5 & \bZ_2^5 &  & \bZ_{(2)}^5  && &  & \bZ_{(2)}^5  & \bZ_2^5 & \bZ_2^5 &  & \bZ_{(2)}^5 &  &  &  \\
& & \bZ_2 & & & \bZ_4 & \bZ_2 & \bZ_2 && \bZ_{(2)}& & \bZ_2 & \bZ_4 & \bZ_2 & \bZ_2 & &   \\
\hline\hline
d & 128&129&130&131&132&133&134&135&136&137&138&139&140&141&142&143\\
\hline
\pi_d(\tmf)_{(2)}& \bZ_{(2)}^6 & \bZ_2^6 & \bZ_2^6 &  & \bZ_{(2)}^6  && &  & \bZ_{(2)}^6  & \bZ_2^6 & \bZ_2^6 &  & \bZ_{(2)}^6 &  &  &  \\
& \bZ_2 & \bZ_2 & \bZ_4 & \bZ_2 & & &  & \bZ_2 & \bZ_2 & \bZ_2 & \bZ_2 &&&&\bZ_2 \\
\hline\hline
d & 144&145&146&147&148&149&150&151&152&153&154&155&156&157&158&159\\
\hline
\pi_d(\tmf)_{(2)}& \bZ_{(2)}^6 & \bZ_2^6 & \bZ_2^6 &  & \bZ_{(2)}^6  && &  & \bZ_{(2)}^7  & \bZ_2^7 & \bZ_2^7 &  & \bZ_{(2)}^7 &  &  &  \\
& \bZ_{(2)} && & \bZ_8 & \bZ_2 & \bZ_2 & \bZ_8 & &  & \bZ_2 & & \bZ_2 & \bZ_2\\
\hline\hline
d & 160&161&162&163&164&165&166&167&168&169&170&171&172&173&174&175\\
\hline
\pi_d(\tmf)_{(2)}& \bZ_{(2)}^7 & \bZ_2^7 & \bZ_2^7 &  & \bZ_{(2)}^7  && &  & \bZ_{(2)}^7  & \bZ_2^7 & \bZ_2^7 &  & \bZ_{(2)}^7 &  &  &  \\
& & \bZ_2 & \bZ_2 && \bZ_2 &&&&\bZ_{(2)} \\
\hline\hline
d & 176&177&178&179&180&181&182&183&184&185&186&187&188&189&190&191\\
\hline
\pi_d(\tmf)_{(2)}& \bZ_{(2)}^8 & \bZ_2^8 & \bZ_2^8 &  & \bZ_{(2)}^8  && &  & \bZ_{(2)}^8  & \bZ_2^8 & \bZ_2^8 &  & \bZ_{(2)}^8 &  &  &  \\
& \\
\hline\hline
\end{array}
\]
\caption{Table of $\pi_d(\tmf)_{(2)}$.  
For each $d$ it is a direct sum of the entries on the first row and the second row. 
The second row is periodic with period 192.
\label{table:2}
}
\end{table}

\begin{table}
\[
\begin{array}{c|ccccccccccccccccccccc}
\hline\hline
d &  0 & 1 & 2 & 3 & 4 & 5 & 6 & 7 & 8 & 9 & 10 & 11 & 12  & 13 & 14 & 15  \\
\hline
\pi_d(\tmf)_{(3)}& \bZ_{(3)} & &  &  &  && &  & \bZ_{(3)}  &  &  &  & \bZ_{(3)} &  &  &  \\
&   & & & \bZ_3 & & & & &&& \bZ_3 & &&\bZ_3 \\
\hline\hline
d &  16 & 17 & 18 & 19 & 20 & 21 & 22 & 23 & 24 & 25 & 26 & 27 & 28  & 29 & 30 & 31  \\
\hline
\pi_d(\tmf)_{(3)}& \bZ_{(3)} & &  &  & \bZ_{(3)} && &  & \bZ_{(3)}  &  &  &  & \bZ_{(3)} &  &  &  \\
&  &&&&\bZ_3 &&&&\bZ_{(3)}&&&\bZ_3 &&& \bZ_3  \\
\hline\hline
d &  32&33&34&35&36&37&38&39&40&41&42&43&44&45&46&47\\
\hline
\pi_d(\tmf)_{(3)}& \bZ_{(3)}^2 & &  &  & \bZ_{(3)}^2 && &  & \bZ_{(3)}^2  &  &  &  & \bZ_{(3)}^2 &  &  &  \\
&   &&&&&\bZ_3 &&&\bZ_3\\
\hline\hline
d & 48&49&50&51&52&53&54&55&56&57&58&59&60&61&62&63\\
\hline
\pi_d(\tmf)_{(3)}& \bZ_{(3)}^2 & &  &  &\bZ_{(3)}^2  && &  & \bZ_{(3)}^3  &  &  &  & \bZ_{(3)}^3 &  &  &  \\
&   \bZ_{(3)}\\
\hline\hline
d & 64&65&66&67&68&69&70&71\\
\hline
\pi_d(\tmf)_{(3)}& \bZ_{(3)}^3 & &  &  & \bZ_{(3)}^3 && &  &  \\
&\\
\hline\hline
\end{array}
\]
\caption{Table of $\pi_d(\tmf)_{(3)}$.  
For each $d$ it is a direct sum of the entries on the first row and the second row. 
The second row is periodic with period 72.
\label{table:3}}

\end{table}

$\pi_{n\ge 0}(\Tmf)$  is equal to $\pi_n(\tmf)$,
$\pi_{-1}(\Tmf)$ to $\pi_{-20}(\Tmf)$ are all zero, 
and $\pi_{n\le -21}(\Tmf)$ is determined by the fact that 
$\Tmf$ is self Anderson dual, $I_\bZ\Tmf\simeq\Sigma^{21}\Tmf$ \cite{Sto1,Sto2}.
In particular we have the exact sequence \begin{equation}
0\to \mathrm{Ext}(\pi_{n-1}(\Tmf),\bZ)\to \pi_{-21-n}(\Tmf)\to \Hom(\pi_n(\Tmf),\bZ)\to 0,
\end{equation} meaning that \begin{equation}
\pi_{-21-n}(\Tmf) \simeq \mathrm{Tors} (\pi_{n-1}(\Tmf)) \oplus \mathrm{Free} (\pi_n(\Tmf)),
\end{equation}
although non-canonically.

\section{Twists of $\KO$ and $\TMF$ and the Segal-Stolz-Teichner conjecture}
\label{sec:twists}
Our main Theorem~\ref{thm_vanishing}
 established the fact that there is no anomalies whatsoever in perturbative heterotic compactifications,
under the assumption that the conjecture of Segal, Stolz, Teichner is valid.
In the physics discussion leading to the formulation of the statement of the theorem, 
we also needed to assume that the twists of $\TMF$ have a certain form suggested by the conjecture.
This appendix is to show that it is indeed the case.

The study of twists of $\K$ and $\KO$ goes back to \cite{DonovanKaroubi}, where it was shown that $\KO^\bullet(X)$ can be twisted by elements of $\H^1(X,\Z2)\times \H^2(X,\Z2)$.
A more modern analysis of twists of generalized cohomology theories  was given in \cite{ABG}
and was applied to $\K$, $\KO$ and $\TMF$ there,
where it was shown that $\KO^\bullet(X)$ and $\TMF^\bullet(X)$ can be twisted by elements of 
$[X,\BO\langle0,1,2\rangle]$ and $[X,\BO\langle0,\ldots,4\rangle]$ respectively.
Here, for a path-connected space $Y$ and a positive integer $n$, we denote by $Y\langle 0, \ldots, n\rangle$ the $n$-stage Postnikov system, i.e.~a path-connected space equipped with a continuous map
\begin{align}
    p \colon Y \to Y\langle 0, \ldots, n\rangle, 
\end{align}
so that $\pi_{> n} (Y) = 0$ and we have isomorphisms $\pi_k(Y) \stackrel{p_*}{\simeq} \pi_k(Y\langle 0, \ldots, n\rangle)$ for $0 \le k \le n$. 

More precisely in the case of $\mathrm{TMF}$, the construction in \cite[Section 8]{ABG}, applied to $\mathrm{TMF}$ rather than $\mathrm{tmf}$, gives a map\footnote{
Apply the $(\Sigma^\infty_+ \Omega^\infty, gl_1)$-adjunction to the map $\Sigma_+^\infty F \to MT\mathrm{String}$ in \cite[Remark 8.4]{ABG}. 
}
\begin{align}\label{eq_ABG_TMF}
    \BO \langle0,\ldots,4\rangle \to B\GL_1 MT\mathrm{String},  
\end{align}
by which an element in $[X, \BO \langle0,\ldots,4\rangle]$ induces a twist of $\Omega^{\mathrm{string}}$ on $X$. 
By composing \eqref{eq_ABG_TMF} with $\mathrm{Wit}_{\mathrm{string}} \colon B\GL_1 MT\mathrm{String} \to B\GL_1 \mathrm{TMF}$, $\sigma \colon \mathrm{TMF} \to \mathrm{KO}((q))$ and with $B\GL_1 R \to \mathrm{Aut}(I_\bZ R)$, 
an element in $[X, \BO \langle0,\ldots,4\rangle]$ also induces twists of $\mathrm{TMF}$, $\mathrm{KO}((q))$ and their Anderson duals. 
$MT\mathrm{String}$-module homomorphisms between these spectra, such as $\alpha_{\mathrm{string}} \colon \mathrm{TMF} \to \Sigma^{-20}I_\bZ MT\mathrm{String}$, induce the corresponding transformation on those twisted theories with twists coming from a common element in $[X, \BO \langle0,\ldots,4\rangle]$. 
We have the corresponding statement for the case of $\mathrm{KO}$, where we use 
\begin{equation}
\BO\langle0,1, 2\rangle \to B\GL_1 MT\Spin
\end{equation}
 and the morphism $\mathrm{ABS} \colon B\GL_1 MT\Spin \to B\GL_1 \KO$ induced by the Atiyah-Bott-Shapiro orientation.

Theorem \ref{thm:SST} and Conjecture \ref{conj:SST} of Segal, Stolz and Teichner posit that $\KO(X)$ and $\TMF(X)$ classify
1-dimensional unitary \Nequals1 supersymmetric quantum field theories with \pinm\ structure
and
2-dimensional unitary \Nequals{(0,1)} supersymmetric quantum field theories systems with spin structure
up to continuous deformations, respectively.
From this perspective, it is natural to identify the twists of $\KO(X)$ and $\TMF(X)$
with the anomalies of respective systems parameterized over $X$, as was already mentioned in \cite{Gukov:2018iiq,Johnson-Freyd:2020itv}.
According to \PhysicsFact~\ref{fact:anomaly}, they are respectively given by \begin{equation}
(I_\bZ\Omega^\pinm )^3(X)
\mbox{ and }
(I_\bZ\Omega^\spin)^4(X). 
\end{equation}
If we take a basepoint in $X$ we have $(I_\bZ\Omega^\pinm )^3(X)
= \Z8 \oplus (\widetilde{I_\bZ\Omega^\pinm })^3(X)$ and $(I_\bZ\Omega^\spin)^4(X)
= \bZ \oplus (\widetilde{I_\bZ\Omega^\spin})^4(X)$, and the parts $\Z8$ and $\bZ$ have already been identified with the degrees of $\KO$ and $\TMF$ as part of Theorem \ref{thm:SST} and Conjecture \ref{conj:SST}, and therefore it is natural to suppose the following Proposition \ref{prop_twist}:

\begin{prop}\label{prop_twist}
For any CW-complex $X$, we have a natural isomorphism
\begin{align}
    [X, \bZ/8\bZ \times \BO\langle0,1,2\rangle] &\simeq ({I_\bZ\Omega^\pinm })^3(X)
    \label{eq_twist_KO},\\
    [X,\bZ \times \BO\langle0,\ldots,4\rangle] &\simeq ({I_\bZ\Omega^\spin})^4(X), 
    \label{eq_twist_TMF}
\end{align}
which fits into the following commutative diagram. 
\begin{align}\label{diag_twist}
    \xymatrix{
    [X, \bZ/8\bZ \times \BO\langle0,1,2\rangle] \ar[r]^-{\simeq}& ({I_\bZ\Omega^\pinm })^3(X) \\
    [X,\bZ \times \BO\langle0,\ldots,4\rangle] \ar[r]^-{\simeq} \ar[u]& ({I_\bZ\Omega^\spin})^4(X) \ar[u]^-{I_\bZ \alpha}.  
    }
\end{align}
Here the left vertical arrow is induced by the canonical map $\bZ \times \BO\langle0,\ldots,4\rangle \to \bZ/8\bZ \times \BO\langle0,1,2\rangle $, and the right vertical arrow is the Anderson dual to the natural transformation $\alpha$ given in Definition~\ref{def:alpha} below.
\end{prop}

\begin{defn}
\label{def:alpha}
The natural transformation $\alpha$ is given by the following:
\begin{align}\label{eq_def_alpha}
    \alpha \colon \Omega^\pinm _*(X) &\to \Omega^\spin_{*+1}(X) \\
    [f \colon M \to X] &\mapsto [f \colon S^1 \times_{\Z{2}} \widetilde{M} \to X], \notag
\end{align}
Here, for a closed \pinm\ manifold $M$ we denote by $\widetilde{M}$ its orientation double cover with the induced spin structure, and equip $S^1$ with the nontrivial spin structure. 
Fix a diffeomorphism $S^1 \simeq \mathrm{U}(1)$. 
The $\Z{2}$-action on $S^1 \times \widetilde{M}$ is given by $(x, y) \mapsto (\bar{x}, y')$, where $y'$ is the other point in the same fiber of $\widetilde{M} \to M$ as $y$. 
\end{defn}

Before proceeding to the proof, we mention that the Pontryagin dual of the spin bordism groups,
$\Hom(\Omega^\spin_d(X),\Q/\bZ)$,
was determined for $d\le 3$  in \cite{BF3} and for $d=4$ in \cite{BF4}.
Their results for $d=3$ is closely related to our  \eqref{eq_twist_TMF}.

\begin{proof}[Proof of Proposition \ref{prop_twist}]
The statement is equivalent to the claim that $\Z{8}\times \BO\langle0,1,2\rangle$ and $\bZ \times \BO\langle0,\ldots,4\rangle$ have the homotopy types of the third and the fourth space of the $\Omega$-spectra representing $(I_\bZ\Omega^\pinm )^\bullet$ and $(I_\bZ\Omega^\spin)^\bullet$, respectively. 
For a spectrum $E$ and a nonnegative integer $n$, we denote by $E_n$ its $n$-th space. 

First we prove \eqref{eq_twist_TMF}. 
Let $\mathrm{ABS}: MT\Spin\to \KO$ be the Atiyah-Bott-Shapiro orientation
and $\gamma_\KO: \KO \to \Sigma^4I_\bZ\KO$  be the self-Anderson-duality of $\KO$ theory.
We then have a transformation of generalized cohomology theory $I_\bZ \mathrm{ABS} \circ \gamma_{\KO} \colon \KO^\bullet \to (I_\bZ\Omega^\spin)^{\bullet + 4}$ by combining them.
Taking an $\Omega$-spectrum representing $(I_\bZ\Omega^\spin)^{\bullet}$
and using $\KO_0\sim \bZ\times \BO$,
it determines an element which appeared in Remark \ref{rem_ABS}, 
\begin{align}\label{eq_proof_twist_TMF}
   I_\bZ \mathrm{ABS} \circ \gamma_{KO}  \in [ \bZ \times \BO, (I_\bZ\Omega^\spin)_4].
\end{align}
For any $k > 4$, we have $\pi_{k}((I_\bZ\Omega^\spin)_4) = (I_\bZ\Omega^\spin)^{4-k}(\mathrm{pt}) = 0$. 
Thus it is enough to show that the map \eqref{eq_proof_twist_TMF} induces isomorphisms on $\pi_{k}$ for $0 \le k \le 4$. 
Then the isomorphism \eqref{eq_twist_TMF} is given by \eqref{eq_proof_twist_TMF}. 
To show it, consider the following commutative diagram, 
\begin{align}\label{diag_proof_twist_TMF}
    \xymatrix{
    0 \ar[r] & \mathrm{Ext}(\KO_{3-k}(\pt), \bZ) \ar[r] \ar[d]^-{\mathrm{ABS}} & (I_\bZ \KO)^{4-k}(\pt) \ar[r] \ar[d]^-{I_\bZ \mathrm{ABS}} & \Hom(\KO_{4-k}(\pt), \bZ) \ar[r] \ar[d]^-{\mathrm{ABS}} & 0 \\
   0 \ar[r] & \mathrm{Ext}(\Omega^\spin_{3-k}(\pt), \bZ) \ar[r] & (I_\bZ\Omega^\spin)^{4-k}(\pt) \ar[r] & \Hom(\Omega^\spin_{4-k}(\pt), \bZ) \ar[r] & 0, 
    }
\end{align}
where the rows are exact. 
Since the right and the left vertical arrows are isomorphisms for each $0 \le k \le 4$, by the five lemma (or more simply just noting that the right or the left groups are zero for each $k$), we see that the middle arrow is also an isomorphism. 
Composing it with $\gamma_{\KO} \colon \KO^{-k}(\pt) \simeq (I_\bZ \KO)^{4-k}(\pt)$, we see that \eqref{eq_proof_twist_TMF} induces isomorphisms on $\pi_{k}$ for $0 \le k \le 4$. 
So we get \eqref{eq_twist_TMF}. 

Next we prove \eqref{eq_twist_KO} and the commutativity of \eqref{diag_twist}. 
Let us take an $\Omega$-spectrum representing $(I_\bZ\Omega^\pinm )^{\bullet}$, and consider the following composition. 
\begin{align}\label{eq_proof_twist_KO}
     I_\bZ\alpha \circ I_\bZ \mathrm{ABS} \circ \gamma_{\KO}  \colon  \KO_0 \sim \bZ \times \BO \to (I_\bZ\Omega^\spin)_4 \to (I_\bZ\Omega^\pinm )_3. 
\end{align}
For any $k > 3$, we have $\pi_{k}((I_\bZ\Omega^\pinm )_3) = (I_\bZ\Omega^\pinm )^{3-k}(\mathrm{pt}) = 0$. 
Also for $k = 3$ we have $\pi_{3}((I_\bZ\Omega^\pinm )_3) = (I_\bZ\Omega^\pinm )^{0}(\mathrm{pt}) = \Hom(\Omega^\pinm _{0}(\pt), \bZ) = 0$. 
Thus, it is enough to show that the homomorphisms induced on $\pi_k$ by the map \eqref{eq_proof_twist_KO} are isomorphisms for $k =1, 2$ and coincides with the quotient map $\bZ \to \Z{8}$ for $k = 0$, and define the isomorphism \eqref{eq_twist_KO} by \eqref{eq_proof_twist_KO}. Then the commutativity of \eqref{diag_twist} follows directly by the construction. 
Since we already know that the map \eqref{eq_twist_TMF} induces isomorphisms of these degrees of the homotopy groups, it is enough to show the corresponding statement for the map $I_\bZ \alpha \colon (I_\bZ\Omega^\spin)_4 \to (I_\bZ\Omega^\pinm )_3$.  

We have, accoring to \cite{KirbyTaylor1990},
\begin{align}
    \Omega^\pinm _0(\pt) &= \Z{2}, & \Omega^\pinm _1(\pt) &= \Z{2}, & \Omega^\pinm _2(\pt) &= \Z{8}, & \Omega^\pinm _3(\pt) &= 0; \label{eq_pin-_bordism_grp}
\end{align}
It is also classic that the low-dimensional spin bordism groups are given by 
\begin{align}
    \Omega^\spin_1(\pt) &= \Z{2}, & \Omega^\spin_2(\pt) &= \Z{2}, & \Omega^\spin_3(\pt) &= 0, & \Omega^\spin_4(\pt) &= \bZ. \label{eq_spin_bordism_grp}
\end{align}
By a straightforward check on the generators, we see that $\alpha$ gives isomorphisms $\alpha \colon \Omega^\pinm _\bullet(\pt) \to \Omega^\spin_{\bullet+1}(\pt)$ for $\bullet = 0, 1$. 
Using the commutative diagram corresponding to \eqref{diag_proof_twist_TMF}, we see that $I_{\bZ}\alpha$ induces isomorphisms on $\pi_k$ for $k = 1, 2$. 
The statement for $\pi_0$ follows by Corollary \ref{cor_pin-_spin_dual} below in the appendix.
This completes the proof. 

\end{proof}

\section{Some examples of $(I_\bZ \Omega^{\mathcal{B}})^{\bullet}(\pt) \to (I_\bZ \Omega^{\mathcal{B}'})^{\bullet+n}(\pt)$}
\label{sec:last}

Here we determine the homomorphism $(I_\bZ \Omega^{\mathcal{B}})^{\bullet}(\pt) \to (I_\bZ \Omega^{\mathcal{B}'})^{\bullet+n}(\pt)$ in a few cases.
They are used in other parts of the paper.

We start with a general setting as follows. 
Suppose we have a morphism of spectra $f \colon E \to E'$. 
Then its {\it mapping cone} is a spectrum $C$ equipped with a morphism $E' \to C$ so that
\begin{align}
    E \xrightarrow{f} E' \to C
\end{align}
is an exact triangle in the stable homotopy category. 
In particular, it produces long exact sequences for the corresponding generalized homology theories and cohomology theories, 
\begin{align}
    \cdots \to E_d(X) \xrightarrow{f} E'_d(X) \to &C_d(X) \to E_{d-1}(X) \xrightarrow{f} E'_{d-1}(X) \to \cdots, \  
    \label{eq_exact_cone}\\
     \cdots \to E^d(X) \xrightarrow{f} E'^d(X) \to &C^d(X) \to E^{d+1}(X) \xrightarrow{f} E'^{d+1}(X) \to \cdots .   
      \notag
\end{align}
Taking the Anderson duals, $I_\bZ C \to I_\bZ E' \xrightarrow{I_\bZ f} I_\bZ E$ is also an exact triangle. 

\begin{lem}\label{lem_cone_dual}
In the above settings, let $d$ be an integer and $k$ be a positive integer. 
Assume that we have
\begin{align}\label{eq_cone_condition}
    E_d(\pt) = 0, \ E_{d-1}(\pt) \simeq \Z{k},\  E'_{d}(\pt) \simeq \bZ,\  E'_{d-1}(\pt) = 0, \ \mbox{and } C_d(\pt) \simeq \bZ. 
\end{align}
Then we have $(I_\bZ E)^d(\pt) \simeq \mathrm{Ext}(E_{d-1}(X), \bZ) \simeq \Z{k}$ and $(I_\bZ E')^d(\pt) \simeq \mathrm{Hom}(E'_d(X), \bZ) \simeq \bZ$, and the homomorphism 
\begin{align}
    I_\bZ f \colon (I_\bZ E')^d(\pt) \simeq \bZ \to (I_\bZ E)^d(\pt) \simeq \Z{k}
\end{align}
maps a generator of $\bZ$ to a generator of $\Z{k}$. 
\end{lem}
\begin{proof}
By \eqref{eq_cone_condition}, the exact sequence \eqref{eq_exact_cone} becomes $0 \to \bZ \xrightarrow{\times k} \bZ \to \Z{k} \to 0$ for an appropriate choice of generators. 
Consider the following commutative diagram:
\begin{align*}
\hskip-.5em    \xymatrix{
    0 \ar[r] & \mathrm{Ext}(C_{d-1}(\pt), \bZ)  \ar[r] \ar[d] & (I_\bZ C)^{d}(\pt)  \ar[r] \ar[d] & \Hom(C_{d}(\pt), \bZ) \simeq \bZ \ar[r] \ar[d]^-{\times k} & 0 \\
   0 \ar[r] & \mathrm{Ext}(E'_{d-1}(\pt), \bZ) = 0 \ar[r] \ar[d]^-{f} & (I_\bZ E')^{d}(\pt) \simeq \bZ \ar[r]^-{\simeq} \ar[d]^-{I_\bZ f} & \Hom(E'_{d}(\pt), \bZ)  \simeq \bZ \ar[r] \ar[d]^-{f} & 0 \\ 
   0 \ar[r] & \mathrm{Ext}(E_{d-1}(\pt), \bZ) \simeq \Z{k} \ar[r]^-{\simeq} & (I_\bZ E)^{d}(\pt) \simeq \Z{k} \ar[r] & \Hom(E_{d}(\pt), \bZ) = 0 \ar[r] & 0 
    }
\end{align*}
Here we have used the exactness of the rows \eqref{eq_exact_IE}. 
Since the middle column is also exact, we get Lemma \ref{lem_cone_dual}. 
\end{proof}

For us $E$ and $E'$ are some bordism theories. 
Here are examples. 

\begin{ex}
Consider the case where $E= MT\mathrm{String}$, $E' =MT\Spin$, $\iota \colon MT\mathrm{String} \to MT\Spin$ is the forgetful map.
In this case the mapping cone of $\iota$ is the {\it relative bordism theory} $\Omega^\iota_* = C_*$ with respect to $\iota$, see \cite[p.25--26]{Stong1968}. 
An element of $\Omega^\iota_d(X)$ is represented by a pair $(W^{d}, M^{d-1})$, where $M^{d-1}$ is a $(d-1)$-dimensional closed manifold equipped with a string structure with a map to $X$, and $W^{d}$ is a $d$-dimensional compact manifold with a spin structure and a map to $X$, which bounds $M$ as a spin manifold with a map to $X$. 
The group $\Omega^\iota_d(X)$ is defined to be the group consisting of the bordism classes $[W, M]$ of such pairs. 
Let us consider the case $d=4$.
We know that
\begin{align}
    \Omega^{\mathrm{string}}_4(\pt) = 0, \ \Omega^{\mathrm{string}}_3(\pt) = \Z{24}, \ 
    \Omega^{\mathrm{spin}}_4(\pt) = \bZ, \ 
    \Omega^{\mathrm{spin}}_3(\pt) = 0. 
\end{align}
To check the condition \ref{eq_cone_condition}, we need to show $\Omega^\iota_4(\pt) =\bZ$. 
To see this, we note that the $K3$-surface admits a framing away from $24$ points, so that the induced framing on the boundary $S^3 = \partial D_i^4$ of a disk neighborhood $D_i^4$ ($1 \le i \le 24$) of each point is isomorphic to the Lie group framing on $\mathrm{SU}(2) \simeq S^3$.\footnote{%
A particularly nice, concrete way to see this was discussed in \url{https://mathoverflow.net/a/58263/5420} by T.~Mrowka, who attributes the argument to M. Atiyah;
 the authors learned this tidbit  from Justin Kaidi.
The proof goes as follows. 
As the Euler number of $K3$ is 24, one can pick a vector field $X$ on it with 24 isolated zeros of index 1. 
One can also introduce a hyperk\"ahler metric on $K3$.
Then, the framing away from these 24 points is explicitly given by 
$(X,(\iota_X \omega_1)^*,(\iota_X \omega_2)^*,(\iota_X \omega_3)^*)$,
where $\omega_{1,2,3}$ are the three self-dual 2-forms coming from the hyperk\"ahler structure.
} 
This means that we have the following equation in $\Omega^\iota_{4}(\pt)$, 
\begin{align}
    [K3, \varnothing] = 24[D^4, S^3] + [K3 \setminus \sqcup_{i = 1}^24 D_i^4, \sqcup_{i = 1}^24 (-S^3)] = 24[D^4, S^3]. 
\end{align}
Notice that the element $[K3, \varnothing] \in \Omega^\iota_4(\pt)$ is the image of the generator $[K3] \in \Omega^\spin_4(\pt)$, and the element $[D^4, S^3]\in \Omega^\iota_{4}(\pt)$ maps to the generator $[S^3] \in \Omega^\text{string}_3(\pt)$. 
Since we know that the group $\Omega^\iota_4(\pt)$ fits into the exact sequence \eqref{eq_exact_cone}, we conclude that $[D^4, S^3]$ generates $\Omega^\iota_4(\pt)$ and it is isomorphic to $\bZ$. 
Thus, we can apply Lemma \ref{lem_cone_dual} to this setting and obtain the following corollary,
which was used in  Sec.~\ref{sec:implication}.
\begin{cor}\label{cor_string_spin_dual}
The Anderson dual to the forgetful homomorphism $\iota \colon M\mathrm{String} \to M\Spin$ in degree four, 
\begin{multline*}
    I_\bZ \iota \colon (I_\bZ \Omega^{\mathrm{spin}})^4(\pt) \simeq \mathrm{Hom}(\Omega^{\mathrm{spin}}_4(\pt), \bZ) \simeq \bZ \\
    \to (I_\bZ \Omega^{\mathrm{string}})^4(\pt) \simeq \mathrm{Ext}(\Omega^{\mathrm{string}}_3(\pt), \bZ) \simeq \Z{24}
\end{multline*}
maps a generator to a generator. 
\end{cor}
\end{ex}

\begin{ex}
Consider the case where $E = MT\Pinm$, $E' = \Sigma^{-1} MT\Spin $, and $\alpha \colon MT\Pinm \to  \Sigma^{-1} MT\Spin$ is given by \eqref{eq_def_alpha}.
Also in this case, the mapping cone of $\alpha$ is given by the relative bordism theory $\Omega^\alpha_* = C_*$ with respect to $\alpha$.
A straightforward modification of the argument in \cite[p.25--26]{Stong1968} shows that the group $\Omega^\alpha_d(X)$ is the set of bordism classes $[W^{d+1}, M^{d-1}]$ of pairs $(W^{d+1}, M^{d-1})$, where $M^{d-1}$ is a $(d-1)$-dimensional closed manifold equipped with a \pinm\ structure with a map to $X$, and $W^{d+1}$ is a $(d+1)$-dimensional compact manifold with a spin structure with a map to $X$, which bounds the spin manifold $\widetilde{M} \times_{\Z{2}}S^1$ along with the map to $X$. 
We now consider the case $d=3$.
We need to check the condition \ref{eq_cone_condition}. 
Recall we know \eqref{eq_pin-_bordism_grp} and \eqref{eq_spin_bordism_grp}. 

\begin{lem}\label{lem_spin_pin-_relative}
We have an isomorphism $\Omega_3^\alpha(\pt) \simeq \bZ$ so that the exact sequence \eqref{eq_exact_cone} for $d=3$ and $X = \pt$ becomes
\begin{align}\label{diag_claim_relative}
    \xymatrix{
      \Omega^\pinm _{3}(\pt) \ar[r]^-{\alpha}\ar@{=}[d] &  \Omega^\spin_{4}(\pt) \ar[r]\ar[d]^-{\simeq} &  \Omega^\alpha_{3}(\pt) \ar[r]\ar[d]^-{\simeq} & \Omega^\pinm _{2}(\pt) \ar[r]^-{\alpha}\ar[d]^-{\simeq} & \Omega^\spin_{3}(\pt) \ar@{=}[d] \\
    0 \ar[r] & \bZ \ar[r]^-{\times 8} & \bZ \ar[r] & \Z{8} \ar[r] & 0. 
    }
\end{align}
Here we use the generator $[K3] \in \Omega^\spin_{4}(\pt)$ and $[\bR \mathbb{P}^2] \in \Omega^\pinm _{2}(\pt)$. 
\end{lem}
\begin{proof}
Recall the Kummer construction of the $K3$-surface,
where it is obtained by considering $T^4 / (\Z{2})$ ($\Z{2}$ acting on $T^4 = (S^1)^4$ by the reflection of each component) which has $16$ singular points, and blowing up each singularity. 
Now regard $T^4 = T^3 \times S^1$. 
Take the $\epsilon$-balls $D^3_i$, $i = 1, \cdots, 8$, centered at each of the fixed points of the $\Z{2}$-action on $T^3$. 
Then we have an inclusion 
\begin{align}\label{eq_proof_claim_relative_1}
    (T^3 \setminus \sqcup_{i = 1}^8 D_i^3) \times_{\Z{2}}S^1 \subset K3
\end{align}
with boundary
\begin{align}\label{eq_proof_claim_relative_2}
    \partial \left( (T^3 \setminus \sqcup_{i = 1}^8 D_i^3) \times_{\Z{2}}S^1\right) = \sqcup_{i = 1}^8 (-S^2) \times_{\Z{2}} S^1 = \sqcup_{i = 1}^8 \widetilde{(-\bR\mathbb{P}^2)} \times_{\Z{2}} S^1. 
\end{align}
Here the \pinm\ manifold $(T^3 \setminus \sqcup_{i = 1}^8 D_i^3) /(\Z{2})$ bounds the \pinm\ manifold $\sqcup_{i=1}^8 (-\bR\mathbb{P}^2)$, and the spin structures on \eqref{eq_proof_claim_relative_1} and \eqref{eq_proof_claim_relative_2} are related to these \pinm\ structures by the map \eqref{eq_def_alpha}. 
This means that we have
\begin{align}\label{eq_proof_claim_relative_3}
    [ (T^3 \setminus \sqcup_{i = 1}^8 D_i^3) \times_{\Z{2}}S^1, \sqcup_{i=1}^8 (-\bR\mathbb{P}^2)] = 0 \in \Omega^\alpha_3(\pt). 
\end{align}
The complement of the inclusion \eqref{eq_proof_claim_relative_1} is $8$ copies of a spin manifold $W^4$ with boundary $\widetilde{\bR\mathbb{P}^2} \times_{\Z{2}} S^1$. 
This means that we have an element
\begin{align}
    [W, \bR\mathbb{P}^2] \in \Omega_3^\alpha(\pt), 
\end{align}
which satisfies
\begin{align}\label{eq_proof_claim_relative_4}
[K3, \varnothing] = 
    8[W, \bR\mathbb{P}^2] + [ (T^3 \setminus \sqcup_{i = 1}^8 D_i^3) \times_{\Z{2}}S^1, \sqcup_{i=1}^8 (-\bR\mathbb{P}^2)]
    = 8[W, \bR\mathbb{P}^2] ,  
\end{align}
where the last equality follows by \eqref{eq_proof_claim_relative_3}. 
Notice that the element $[K3, \varnothing] \in \Omega^\alpha_3(\pt)$ is the image of the generator $[K3] \in \Omega^\spin_4(\pt)$, and the element $[W, \bR\mathbb{P}^2]\in \Omega^\alpha_3(\pt)$ maps to the generator $[\bR\mathbb{P}^2] \in \Omega^\pinm _2(\pt)$. 
Since we know that the group $\Omega^\alpha_3(\pt)$ fits into the exact sequence \eqref{eq_exact_cone}, we conclude that $[W, \bR\mathbb{P}^2]$ generates $\Omega^\alpha_3(\pt)$ and is isomorphic to $\bZ$. 
Using this generator, we also get \eqref{diag_claim_relative} and this completes the proof. 
\end{proof}
\end{ex}

By Lemma \ref{lem_spin_pin-_relative}, \eqref{eq_pin-_bordism_grp} and \eqref{eq_spin_bordism_grp}, we get the following corollary, which was used in the proof of Proposition~\ref{prop_twist}:
\begin{cor}\label{cor_pin-_spin_dual}
The Anderson dual to the transformation $\alpha$ in \eqref{eq_def_alpha} in degree three, 
\begin{multline}\label{eq_Ialpha}
    I_\bZ \alpha \colon (I_\bZ\Omega^\spin)^4(\pt) \simeq \Hom(\Omega^\spin_4(\pt), \bZ) \simeq \bZ \\
    \to (I_\bZ\Omega^\pinm )^3(\pt) \simeq \Hom(\Omega^\pinm _2(\pt), \bR/\bZ) \simeq \Z{8}
\end{multline}
maps a generator of $\bZ$ to a generator of $\Z{8}$. 
\end{cor}

\begin{ex}
As a final example, 
consider the case where $E = B\mathrm{SU}(2)_+ \wedge MT\Spin$ and $E' = B\mathrm{SU}(3)_+ \wedge MT\Spin$, $f $ is induced from the inclusion $\mathrm{SU}(2) \hookrightarrow \mathrm{SU}(3)$ and $d = 6$. 
In this case $E_\bullet(\pt) = \Omega^{\mathrm{spin}}_\bullet(B\mathrm{SU}(2))$ and $E_\bullet(\pt) = \Omega^{\mathrm{spin}}_\bullet(B\mathrm{SU}(3))$, and the exact sequence \eqref{eq_exact_cone} is the long exact sequence for the relative spin bordism groups, 
\begin{multline}\label{exact_relative_spin}
     \to  \Omega_d^{\mathrm{spin}}(B\mathrm{SU}(2)) \to \Omega_d^{\mathrm{spin}}(B\mathrm{SU}(3)) \to \Omega_d^{\mathrm{spin}}(B\mathrm{SU}(3), B\mathrm{SU}(2)) \\
     \to \Omega_{d-1}^{\mathrm{spin}}(B\mathrm{SU}(2)) \to \Omega_{d-1}^{\mathrm{spin}}(B\mathrm{SU}(3)) \to . 
\end{multline}

\begin{lem}\label{lem_relative_spin_BSU}
The exact sequence \eqref{exact_relative_spin} for $d = 6$ is  isomorphic to
\begin{align*}
    0 \to \bZ \xrightarrow{\times 2} \bZ \to \Z{2} \to 0. 
\end{align*}
\end{lem}
\begin{proof}
We have a natural transformation $\pi_* \to \Omega^{\mathrm{spin}}_*$ by the forgetful map. 
Moreover for any group $G$ we have $\pi_*(BG) \simeq \pi_{*-1}(G)$. 
Consider the following diagram:{\footnotesize
\begin{align}\label{diag_htpy_to_spin}
\hskip-1em
\vcenter{\xymatrixcolsep{.9em}
    \xymatrix{
      \pi_5(\mathrm{SU}(2))   \ar[r]\ar[d] &  \pi_5(\mathrm{SU}(3)) \ar[r] \ar[d] &  \pi_5(\mathrm{SU}(3), \mathrm{SU}(2))  \ar[r]\ar[d] & \pi_4(\mathrm{SU}(2)) \ar[r]\ar[d] & \pi_4(\mathrm{SU}(3)) \ar[d] \\
       \Omega_6^{\mathrm{spin}}(B\mathrm{SU}(2))  \ar[r] &  \Omega_6^{\mathrm{spin}}(B\mathrm{SU}(3)) \ar[r] &  \Omega_6^{\mathrm{spin}}(B\mathrm{SU}(3), B\mathrm{SU}(2)) \ar[r] & \Omega_{5}^{\mathrm{spin}}(B\mathrm{SU}(2)) \ar[r] & \Omega_{5}^{\mathrm{spin}}(B\mathrm{SU}(3)). 
    }}
\end{align}}
Both rows are exact. 
The first row is isomorphic to the long exact sequence for homotopy groups with respect to the fibration $S^3 \simeq \mathrm{SU}(2) \to \mathrm{SU}(3) \to S^5 \simeq \mathrm{SU}(3) / \mathrm{SU}(2)$. 
Using the result on the homotopy groups of $\mathrm{SU}(3)$ in \cite{MimuraToda1963}, we see that the first row of \eqref{diag_htpy_to_spin} is isomorphic to 
\begin{align}
    \Z{2} \xrightarrow{0} \bZ \xrightarrow{\times 2} \bZ \to \Z{2} \to 0. 
\end{align}
Moreover, from a routine computation using Atiyah-Hirzebruch spectral sequence,  we know that 
\begin{align}
    \Omega_6^{\mathrm{spin}}(B\mathrm{SU}(2)) &= 0 , & \Omega_6^{\mathrm{spin}}(B\mathrm{SU}(3)) \simeq \bZ, \\ 
    \Omega_{5}^{\mathrm{spin}}(B\mathrm{SU}(2)) &= \Z{2}, &
    \Omega_{5}^{\mathrm{spin}}(B\mathrm{SU}(3)) = 0, 
\end{align}
and the first, second, fourth and the fifth vertical arrows in \eqref{diag_htpy_to_spin} are isomorphisms.\footnote{%
The nontrivial ones are the second and the fourth. 
The second one follows e.g.~by combining i) a result in \cite{Bott:1978bw} saying that the Hurewicz 
$\pi_5(\SU(3))\to H_5(\SU(3),\bZ)$ is $\bZ\xrightarrow{\times2}\bZ$,
ii) the  transgression $H_5(\SU(3),\bZ) \simeq H_6(B\SU(3),\bZ)$, and iii) the fact that
$\Omega^\spin_6(B\SU(3))$ in $H_6(B\SU(3),\bZ)$ is an index-2 subgroup as can be seen e.g.~from the Atiyah-Hirzebruch spectral sequence.
The fourth one follows e.g.~by first showing the transgression isomorphism $ \Omega_{5}^{\mathrm{spin}}(B\SU(2)) \simeq\Omega_4^\spin(\SU(2))$, which is $\simeq \Omega_1^\spin(\pt)$,
which is $\simeq \Omega_1^\text{framed}(\pt)\simeq \pi_4(S^3)$.
} 
By the five lemma, we see that the middle vertical arrow is also an isomorphism, and the result follows. 
\end{proof}
By Lemma \ref{lem_relative_spin_BSU}, we can apply Lemma \ref{lem_cone_dual} to this case and get the following. 
\begin{cor}
The pullback by the inclusion $B\mathrm{SU}(2) \to B\mathrm{SU}(3)$, 
\begin{align}
    (I_\bZ \Omega^{\mathrm{spin}})^6(B\mathrm{SU}(3) ) \simeq \bZ \to  (I_\bZ \Omega^{\mathrm{spin}})^6(B\mathrm{SU}(2) ) \simeq \Z{2}
\end{align}
maps a generator to the generator. 
\end{cor}
This last example played a somewhat important role in the development of the study of global anomalies in the physics literature. 
The phenomenon of the global anomaly was originally found in \cite{Witten:1982fp}, 
which in modern terms corresponds to the generator of $(I_\bZ \Omega^{\mathrm{spin}})^6(B\mathrm{SU}(2))\simeq \bZ/2\bZ$.
Before that, only perturbative anomalies, i.e.~the anomalies associated to the free part of the Anderson dual of the bordism groups, were understood,
and therefore it was thought desirable to derive global anomalies from perturbative anomalies.
This was done slightly later in \cite{Elitzur:1984kr}, using this example.
This was recently revisited from a more modern point of view in \cite{Davighi:2020bvi,Davighi:2020uab}.
\end{ex}

\section*{Declarations}
\noindent\emph{Funding:} The research of Yuji Tachikawa is  supported by in part supported  by WPI Initiative, MEXT, Japan through IPMU, the University of Tokyo,
and in part by Grant-in-Aid for JSPS KAKENHI Grant Number 17H04837.
The work of Mayuko Yamashita is supported by Grant-in-Aid for JSPS KAKENHI Grant Number 20K14307 and JST CREST program JPMJCR18T6.
The authors have no other competing interests to declare that are relevant to the content of this article.

\def\arxivfont{\rm}
\bibliographystyle{ytamsalpha}
\bibliography{ref}

\end{document}